\newif\iflipics
\lipicstrue
\lipicsfalse %

\iflipics
\documentclass[english]{lipics-v2021}
\else
\documentclass[11pt]{article}
\fi
\usepackage{hyperref}
\hypersetup{colorlinks=true,allcolors=blue}
\usepackage{amsthm}
\usepackage{amsmath}
\usepackage{amsfonts}
\usepackage{amssymb}
\usepackage{import} 
\usepackage{xifthen}
\usepackage{mathtools}
\usepackage{arydshln}%
\usepackage{xspace}
\usepackage{dsfont}
\usepackage{cleveref} 
\usepackage[cal=esstix]{mathalfa}
\usepackage{algorithmicx}
\usepackage{algorithm}
\usepackage[noend]{algpseudocode}
\newcommand{\AdapRound}[1]{\Statex // Query round #1}

\usepackage{thm-restate}
\usepackage{xcolor} %
\usepackage{tikz}
\usetikzlibrary{fit,calc}

\usepackage{etoolbox}   %
\makeatletter
\patchcmd{\thebibliography}{#1}{A99}{}{}   %
\makeatother

\newcommand{\cut}{\mathcal{C}}

\newcommand{\mintcut}{\mathrm{cut}}

\newcommand{\R}{\mathbb{R}}
\newcommand{\C}{\mathcal{C}}

\newcommand{\Exp}[2]{\mathbb{E}_{ #2}\left[ #1 \right]}

\newcommand{\Probability}[1]{\Pr\left[ #1 \right]}

\newcommand{\tO}{\tilde{O}}

\DeclareMathOperator{\polylog}{polylog}
\providecommand{\set}[1]{{\{#1\}}}

\iflipics
\newtheorem{question}[theorem]{Question}
\authorrunning{Y. Kenneth-Mordoch and R. Krauthgamer}
\author{Yotam Kenneth-Mordoch}{Weizmann Institute of Science, Rehovot, Israel}{yotam.kenneth@weizmann.ac.il}{https://orcid.org/0000-0002-9212-9172}{}
\author{Robert Krauthgamer}{Weizmann Institute of Science, Rehovot, Israel}{robert.krauthgamer@weizmann.ac.il}{https://orcid.org/0009-0003-8154-3735}{}
\relatedversiondetails[linktext={arXiv:2506.20412}]{A full version is available at}{https://arxiv.org/abs/2506.20412}
\Copyright{Y. Kenneth-Mordoch and R. Krauthgamer}

\EventEditors{Anne Benoit, Haim Kaplan, Sebastian Wild, and Grzegorz Herman}
\EventNoEds{4}
\EventLongTitle{33rd Annual European Symposium on Algorithms (ESA 2025)}
\EventShortTitle{ESA 2025}
\EventAcronym{ESA}
\EventYear{2025}
\EventDate{September 15--17, 2025}
\EventLocation{Warsaw, Poland}
\EventLogo{}
\SeriesVolume{351}
\ArticleNo{99}
\else
\newtheorem{theorem}{Theorem}[section]
\newtheorem{question}[theorem]{Question}
\newtheorem{claim}[theorem]{Claim}
\newtheorem{proposition}[theorem]{Proposition}
\newtheorem{lemma}[theorem]{Lemma}
\newtheorem{corollary}[theorem]{Corollary}

\newtheorem{definition}[theorem]{Definition}

\author{
  Yotam Kenneth-Mordoch
  \qquad
  Robert Krauthgamer%
  \thanks{The Harry Weinrebe Professorial Chair of Computer Science.
    Work partially supported by the Israel Science Foundation grant \#1336/23,
    by the Israeli Council for Higher Education (CHE) via the Weizmann Data Science Research Center,
    and by a research grant from the Estate of Harry Schutzman.
  }
  \\ Weizmann Institute of Science
  \\ \texttt{\{yotam.kenneth,robert.krauthgamer\}@weizmann.ac.il}
}
\usepackage[letterpaper, margin=1in]{geometry}

\fi

\title{Cut-Query Algorithms with Few Rounds}

\begin{document}

\maketitle

\begin{abstract}
In the cut-query model, the algorithm can access the input graph $G=(V,E)$
only via cut queries that report, given a set $S\subseteq V$,
the total weight of edges crossing the cut between $S$ and $V\setminus S$.
This model was introduced by Rubinstein, Schramm and Weinberg [ITCS'18]
and its investigation has so far focused on the number of queries
needed to solve optimization problems, such as global minimum cut.
We turn attention to the round complexity of cut-query algorithms,
and show that several classical problems can be solved in this model
with only a constant number of rounds.

Our main results are algorithms for finding a minimum cut in a graph,
that offer different tradeoffs between round complexity and query complexity,
where $n=|V|$ and $\delta(G)$ denotes the minimum degree of $G$:
(i) $\tilde{O}(n^{4/3})$ cut queries in two rounds in unweighted graphs;
(ii) $\tilde{O}(rn^{1+1/r}/\delta(G)^{1/r})$ queries in $2r+1$ rounds for any integer $r\ge 1$ again in unweighted graphs; and
(iii) $\tilde{O}(rn^{1+(1+\log_n W)/r})$ queries in $4r+3$ rounds for any $r\ge1$ in weighted graphs.
We also provide algorithms that find a minimum $(s,t)$-cut
and approximate the maximum cut in a few rounds.
\end{abstract}

\section{Introduction}
\label{sec:introduction}

Graph cuts are a fundamental object in computer science with numerous applications,
in part because they are a basic example of submodular functions.
Recall that a set function $f:2^V\to \R$ is \emph{submodular}
if it satisfies the diminishing-returns property
\begin{equation*}
  \forall A\subset B\subset V, \forall v\not\in B,
  \quad
  f(A\cup \set{v}) - f(A) \ge f(B\cup \set{v}) - f(B) .
\end{equation*}
From the viewpoint of submodular optimization,
it is natural to study the complexity of graph algorithms in the \emph{cut-query model},
which correspond to value queries to the submodular function.
Here, the input graph $G=(V,E,w)$ can be accessed only via oracle queries to its cut function,
namely to $\mintcut_G: 2^V\to \R$, given by 
\begin{equation*}
    \forall S\subseteq V,
    \qquad
    \mintcut_G(S) 
    := \sum_{e\in E:\ |e\cap S|=1} w(e) .
\end{equation*}
Here and throughout $w(e)>0$ is the weight of the edge $e\in E$,
precluding edges of zero weight as they cannot be detected by the algorithm.
As usual, an unweighted graph models unit weights, i.e., $w(e)=1$ for all $e$. 

The \emph{query complexity} of an algorithm is the number of queries
it makes in the worst-case.
While this is often the primary performance measure, 
it is highly desirable that algorithms can be parallelized,
as such algorithms can utilize better the available computing resources. 
The literature on the parallelization of submodular optimization \cite{BS18,EN19,FMZ19, CGJS23},
relies on a measure called \emph{round complexity}, which count the number of sequential rounds of queries the algorithm makes.\footnote{The aforementioned references use the term \emph{adaptivity}, which is the number of rounds minus one.}

\begin{definition}[Round Complexity]
An algorithm has \emph{round complexity} $r$ if it performs $r$ rounds of queries,
meaning that queries in round $t\in [r]$ can depend only on
the answers to queries in previous rounds $1,\ldots,t-1$.
An algorithm is called \emph{non-adaptive} if it uses a single round of queries.
\end{definition}

We study the cut-query model,
from the perspective of round complexity focusing on the problem of finding a global minimum cut.
Previous work~\cite{RSW18,MN20,AEGLMN22,ASW25} shows that
a minimum cut of a graph can be found using
$O(n)$ queries in unweighted graphs and $\tO(n)$ queries in weighted graphs. 
However, it paid no attention to the algorithms' round complexity,
and we fill this gap by studying the tradeoff between round complexity and query complexity.

An analysis of the round complexity of existing algorithms, with some simple modifications, shows that it is possible to find a global minimum cut of a graph using $\tO(n)$ queries with round complexity $\tO(\log^2 n)$ \cite{RSW18,MN20}.
It should be noted that other algorithms do require round complexity $\Omega(n)$ \cite{AEGLMN22,ASW25}.
As another baseline, it is possible to recover the entire graph using $O(n^2)$ deterministic non-adaptive queries \cite{BM11},\footnote{$\tO(|E|)$ queries suffice, however without some bound on the size of the graph the algorithm must make $O(n^2)$ queries in the worst case.} which clearly suffices to find a minimum cut.
This motivates our main question
\begin{question}
    What is the tradeoff between the round complexity and the query complexity of finding a global minimum cut in the cut-query model?
\end{question}

\subsection{Main Results}
We show that the polylogarithmic round complexity of previous work is not necessary for finding a global minimum cut.
We begin by presenting two results for unweighted graphs.
The first one shows that using even two rounds, we can substantially improve on the $O(n^2)$ query complexity of the naive algorithm.
All of our algorithms return a vertex set $S\subseteq V$ and the value of the minimum cut.
In the case of unweighted graphs the algorithms can also report the edges of the cut using one additional round of queries and with the same asymptotic query complexity.
\begin{theorem}[Unweighted Minimum Cut with $2$ rounds]
    \label{theorem:global-minimum-cut-2-rounds}
    Given an unweighted graph $G$ on $n$ vertices, it is possible to find a minimum cut of $G$ using $\tO(n^{4/3})$ cut queries in $2$ rounds.
    The algorithm is randomized and succeeds with probability $1-n^{-2}$.
\end{theorem}

We further provide a smooth tradeoff between the number of queries and the number of rounds, namely $\tO(n^{1+1/r}/\delta(G)^{1/r})$ cut queries with in $2r$ rounds, where $\delta(G)\ge1$ is the minimum degree of $G$.\footnote{We can assume that the graph is connected since connectivity can be checked using $\tO(n)$ cut queries in $1$ round (non-adaptive) \cite{ACK21}.}
\begin{theorem}[Unweighted Minimum Cut with $O(r)$ rounds]
    \label{theorem:min-cut-k-rounds}
    Given an unweighted graph $G$ on $n$ vertices and a parameter $r\in \{1,2,\ldots,\log n\}$, it is possible to find a minimum cut of $G$ using $\tO(rn^{1+1/r}\delta(G)^{-1/r})$ cut queries in $2r+1$ rounds.
    The algorithm is randomized and succeeds with probability $1-n^{-1}$.
\end{theorem}

We also obtain a similar tradeoff also for weighted graphs.
\begin{theorem}[Weighted Minimum Cut Graphs with $O(r)$ rounds]
    \label{theorem:weighted-min-cut}
    Given a weighted graph $G$ on $n$ vertices with integer edge weights bounded by $W$, and a parameter $r\in \{1,2,\ldots,\log n\}$,
    it is possible to find a minimum cut of $G$ using $\tO(rn^{1+(1+\log_n W)/r})$ cut queries in $4r+3$ rounds.
    The algorithm is randomized and succeeds with probability $1-n^{-2}$.
\end{theorem}

Finally, we note that our techniques can also be applied to approximating the maximum cut and finding a minimum $(s,t)$-cut with low adaptivity.
Furthermore, our result for weighted minimum cuts can also be applied to finding a minimum cut in dynamic streams, providing a smooth tradeoff between the number of passes and the storage complexity of the algorithm.
This holds since each cut query in a given round can be computed using $O(\log n)$ storage in a single pass by simply counting the edges in the cut.
Therefore, the round complexity of our algorithm translates immediately to the number of passes required by a streaming algorithm.
These results are detailed in \Cref{sec:additional-results}.

\subsection{Related Work}

\subparagraph*{Algorithms in the cut query model}
In recent years there have been several works on cut problems in the cut-query model, most of which focused on query complexity and not round complexity.
We will discuss the known results for several problems in the cut-query model.

Beginning with finding a minimum cut,  $O(n)$ randomized queries suffice for simple graphs \cite{RSW18,AEGLMN22}.
Meanwhile, finding a minimum cut in a weighted graph requires $\tO(n)$ randomized queries \cite{MN20}.
Finally, in a recent work it was shown that non-trivial deterministic algorithms exist for finding the minimum cut of a simple graph using $\tO(n^{5/3})$ queries \cite{ASW25}.
For the easier problem of determining whether a graph is connected, it is known that $O(n)$ randomized queries suffice \cite{CL23,LC24}.
Additionally, there exists a one round randomized algorithm for connectivity that uses $\tO(n)$ queries \cite{ACK21}.

Moving to the problem of finding a minimum $(s,t)$-cut.
In \cite{RSW18} it was shown that it is possible to compute the minimum $(s,t)$-cut of a simple graph using $O(n^{5/3})$ queries, by showing that using a cut sparsifier it is possible to find a small set of edges that contains the minimum $(s,t)$-cut.
This query complexity was matched in the deterministic setting by \cite{ASW25}.

We also mention the problem of finding a $1-\epsilon$ approximation of the maximum cut.
One way to achieve this is by constructing a cut sparsifier, this was show in \cite{RSW18} for simple graphs and in \cite{PRW24} for weighted graphs, both using $\tO(\epsilon^{-2} n)$ cut queries.
Furthermore, an algorithm for $1/2$-approximation using $O(\log n)$ cut queries in one round and a lower bound  of $\tilde{\Omega}(n)$ on the query complexity needed to find a $1-\epsilon$ approximation of the maximum cut are also known \cite{PRW24}.

Finally, there have been several works on graph problems in different query models, such as additive queries \cite{CK08,BM11,BM12}, quantum cut queries \cite{LSZ21,AL21,AEGLMN22}, matrix multiplication queries \cite{SWYZ19,AL21,AEGLMN22}, Bipartite Independent Set queries \cite{BHRRS20,ACK21,AMM22} and more.

\subparagraph*{Adaptivity in submodular optimization}
In recent years, there has been a growing interest in the adaptivity of algorithms for submodular optimization.
This line of work was initiated by the work of \cite{BS18}, which showed that constrained submodular maximization can be solved within a constant approximation factor in $O(\log n)$ rounds with a polynomial number of queries, the work also provides a matching lower bound.
A string of later works improved on the query complexity and the approximation guarantee, culminating in an algorithm with an optimal $1-1/e$ approximation using $\tO(n)$ queries  \cite{BS18,BBS18,BRS19,EN19,FMZ19,CDK21}.
The round complexity of unconstrained submodular maximization was also studied in \cite{CFK19}.

The adaptivity of algorithms for submodular minimization was also studied in several works.
It was shown that solving the minimum $(s,t)$-cut problem in weighted graphs using $O(\log n)$ rounds of queries requires $\tilde{\Omega}(n^{2})$ queries \cite{ACK19}.
For more general submodular functions, a bound of $\Omega(n/\log n)$ rounds for algorithms with a polynomial number of queries is known \cite{CCK21,BS20,CGJS22}.
Note that all these hard instances are non symmetric, and hence do not apply to the minimum cut problem.
It remains open to study the round complexity of symmetric submodular minimization problems in future work.

\section{Technical Overview}
\label{sec:technical-overview}
Each of our three algorithms for minimum cut is based on a different insight.
Our two-round algorithm uses contractions in the style of~\cite{KT19}
to create a smaller graph that preserves the minimum cut.
Similar contractions have been used in several minimum cut algorithm~\cite{KT19,GNT20,AEGLMN22},
however the existing algorithms aim to reduce the number of vertices in the graph,
while our variant, called $\tau$-star contraction, reduces the number of edges. 

Our $(2r)$-round algorithm for unweighted graphs employs a known approach
of applying a contraction procedure
and then packing (in the contracted graph) edge-disjoint forests.
The main challenge to ensure that forests are edge disjoint when packing them in parallel,
and we overcome it by leveraging an elegant connection to cut sparsification,
which shows that amplifying the sampling probabilities of the sparsifier construction
by a factor of $k$ ensures finding $k$ edge-disjoint forests.

Our algorithm for weighted minimum cut follows the approach of~\cite{MN20},
which constructs a $(1\pm\epsilon)$-cut-sparsifier
and then solves a monotone‑matrix problem that is easily parallelizable. 
Prior work in the cut-query model provided cut sparsifiers \cite{RSW18, PRW24},
but their core step of weight‑proportional edge sampling requires $O(\log n)$ rounds. 
We eliminate this bottleneck with our two-round weight-proportional edge sampling primitive (\Cref{lemma:weighted-edge-sampling}),
and thereby build a $(1\pm \epsilon)$-cut-sparsifier in $3 r$ rounds,
and complete the entire algorithm in $O(r)$ rounds.

A common thread to all our algorithms is the use of edge sampling,
both directly as part of the algorithms
and to construct complex primitives such as forest packings and cut sparsifiers.
In edge sampling, the input is a source vertex $s$ and a target vertex set $T\subseteq V\setminus\{s\}$,
and the goal is to sample an edge from $E(s,T) \coloneqq E\cap (\{s\}\times T)$.
We need two types of edge sampling,
both returning a sampled edge $e\in E(s,T)$ and its weight $w(e)$.
The first one is \emph{uniform edge sampling},
which picks each $e\in E(s,T)$ with equal probability $1/|E(v,T)|$
regardless of its weight, similarly to $l_0$ sampling;
the second one is \emph{weight-proportional edge sampling},
which picks each $e\in E(s,T)$ with probability $w(e)/w(E(v,T))$,
similarly to $l_1$ sampling.
In unweighted graphs these two definitions clearly coincide. 

Weight-proportional edge sampling was used extensively in previous algorithms
for the cut-query model \cite{RSW18,MN20,AEGLMN22,PRW24},
although it had a naive implementation that requires $O(\log n)$ rounds.
This naive implementation, which was proposed in \cite{RSW18},
follows a binary search approach.
Each step partitions the set $T$ at random into two sets $T_1,T_2$ of equal cardinality,
and finds the total weight of the edges in $E(s,T_1)$ and $E(s,T_2)$.
The algorithm then decides whether to continue with $T_1$ or $T_2$
randomly in proportion to their weights.
This approach clearly requires $\Omega(\log n)$ rounds,
which exceeds the $O(1)$ round complexity we aim for,
and furthermore yields only weight-proportional sampling,
while some of our algorithms require uniform sampling in weighted graphs.

\subsection{Edge Sampling Primitives}
Our first technical contribution is a procedure for non-adaptive uniform edge sampling using $O(\log^3 n/\log\log n)$ queries that returns the weight of the sampled edge.
An algorithm for uniform edge sampling using $O(\log^3 n)$ cut queries in one round (i.e. non-adaptively) was proposed in \cite{ACK21}.%
\footnote{The algorithm requires $O(\log^2n \log (1/\delta))$ cut queries and succeeds with probability $1-\delta$, setting $\delta=n^{-2}$ yields the bound.
}
However, their algorithm is based on group testing principles using BIS queries\footnote{A BIS query is given two disjoint sets $A,B\subseteq V$ and returns true if there exists some edge connecting $A,B$ and false otherwise.
}
and returns only the endpoints of the edge, but not its weight, though it is probably not difficult to adapt their algorithm to return the weight as well.
Our algorithm slightly improves on the query complexity of their algorithm, while also returning the edge's weight in the same round.
\begin{lemma}[Uniform Edge Sampling]
    \label{lemma:edge-sampling-1-round}
    Given a (possibly weighted) graph $G$ on $n$ vertices, a source vertex $s\in V$, and a target vertex set $T\subseteq V\setminus\{s\}$, one can return an edge $e$ uniformly sampled from $E(s,T)$, along with its weight $w(e)$, using $O(\log^3 n/\log\log n)$ cut queries in one round, i.e, non-adaptively.
    The algorithm succeeds with probability $1-n^{-4}$.
\end{lemma}
The algorithm works by subsampling the vertices in $T$ into sets $T=T_0\supseteq T_1\supseteq \ldots \supseteq T_{\log n}$, and finds a level $i$ such that $0<|E(s,T_i)| \le O(\log n)$.
It then leverages sparse recovery techniques to recover all the edges in $E(s,T_i)$ and return one of them at random.
To perform this algorithm in one round the algorithm recovers $w(E(s,T_i))$ and performs the queries needed for sparse recovery for all levels in parallel.
It then applies the sparse recovery algorithm to the last level with non-zero number of edges, i.e. the last level with $w(E(s,T_i))>0$.
We note that it is easy to decrease the query complexity of \Cref{lemma:edge-sampling-1-round} by a factor of $O(\log n)$ by adding a round of queries and performing the sparse recovery only on the relevant level.

Our second procedure performs weight-proportional edge sampling in two rounds.
\begin{lemma}[Weight-Proportional Edge Sampling]
    \label{lemma:weighted-edge-sampling}
    Given a graph $G$ on $n$ vertices with integer edge weights bounded by $W$, a source vertex $s\in V$, and a target vertex set $T\subseteq V\setminus\{s\}$, one can return an edge $e$ from $E(s,T)$, along with its weight $w(e)$, sampled with probability $w(e)/w(E(s,T))$ using $O(\log^2n \log^2 (nW))$ queries in two rounds.
    The algorithm succeeds with probability $1-n^{-4}$.
\end{lemma}
The algorithm follows the same basic approach as the uniform edge sampling primitive, i.e. creating nestled subsets $T=T_0\supseteq T_1,\supseteq \ldots \supseteq T_{\log(nW)}$ by subsampling.
However, instead of recovering only the edges that were sampled into the last level with non-zero weight, the algorithm recovers in each level $i$ all edges $e\in E(s,T_i)$ such that $w(e)\approx w(E(s,T))/2^i$.
The recovery procedure is based on a Count-Min data structure \cite{CM05}, which we show can be implemented using few non-adaptive cut queries.
The algorithm then uses another round of queries to learn the exact weights of all edges that were recovered in all the levels.
It then concludes by sampling an edge from the distribution $\{p_e\}$, where for every edge $e$ recovered in the $i$-th level, it sets $p_e \propto 2^i w(e)/w(E(s,T))$;
with probability $1-\sum_e p_e$, the algorithm fails.   
Since the probability of being sampled into the $i$-th level is $2^{-i}$, these probabilities yield the desired distribution over the recovered edges.

To prove that the algorithm returns an edge with high probability, we show that there is a constant probability that the outlined procedure returns some edge $e\in E(s,T)$, and hence repeating the algorithm $O(\log n)$ times yields with high probability an edge sampled from $E(s,T)$ with the desired distribution.
The query complexity of the algorithm is dominated by the Count-Min data structure, which requires $O(\log n \log (nW))$ cut queries in each level.
Applying the recovery procedure on $O(\log (nW))$ levels, and repeating the procedure $O(\log n)$ times yields the stated query complexity.

\subsection{Unweighted Minimum Cut in Two Rounds}
The key idea of our algorithm is to apply an edge contraction technique that produces a graph with substantially fewer edges while preserving each minimum cut with constant probability. 
We then recover the entire contracted graph using a limited number of cut queries by applying existing graph recovery methods using $\tO(m)$ additive queries \cite{CK08,BM11}.
We show that it is possible to simulate $k$ additive queries using $O(n+k)$ cut queries, which allows us to recover a graph with $m\ge \tO(n)$ edges using $O(m)$ cut queries.

\begin{restatable}{corollary}{graphrecoverycorollary}
    \label{corollary:graph-recovery-cut}
    Given a weighted graph $G$ on $n$ vertices and $m$ edges, one can recover the graph $G$ using $O(m + n)$ non-adaptive cut queries.
\end{restatable}

While there has been frequent use of contraction procedures for finding minimum cuts, existing algorithms focused on reducing the number of vertices in the graph \cite{KT19,GNT20,AEGLMN22}, typically to $\tO(n/\delta(G))$.
This does not suffice to guarantee a graph with few edges, for example, applying such a contraction on a graph $G$ with minimum degree $\delta(G)=\log^2 n$ and $\Omega(n)$ vertices of degree $n$, might yield a graph with $\tilde{\Omega}(n^2)$ edges.
In contrast, our goal is to reduce the number of edges in the graph, allowing efficient recovery of the contracted graph using \Cref{corollary:graph-recovery-cut}.
The main technical contribution of this section is a new contraction algorithm, which we call $\tau$-star contraction, that reduces the number of edges to $\tO((n/\tau)^2 +n\tau)$.

Our contraction algorithm is a thresholded version of the star-contraction algorithm introduced in \cite{AEGLMN22}.
The original version of the star-contraction algorithm of \cite{AEGLMN22} can be roughly described as follows.
Sample a subset of \emph{center vertices} $R\subseteq V$ uniformly at random with probability $p=O(\log n/\delta(G))$.
For every $v\in V\setminus R$, choose uniformly a vertex in $r\in N_G(v)\cap R$ and contract the edge $(v,r)$, keeping parallel edges, which yields a contracted multigraph $G'$.
The main guarantee of the algorithm is that if the minimum cut of $G$ is non-trivial, i.e. not composed of a single vertex, then $\lambda(G)=\lambda(G')$ with constant probability, where $\lambda(G)$ denotes the value of a minimum cut in a graph.
The following theorem states this formally.
\begin{theorem}[Theorem 2.2 in \cite{AEGLMN22}]
    \label{theorem:original-star-contraction}
    Let $G=(V,E)$ be an unweighted graph on $n$ vertices with a non-trivial minimum-cut value $\lambda(G)$.
    Then, the star-contraction algorithm yields a contracted graph $G'$ that,
    \begin{enumerate}
        \item has at most $O(n \log n /\delta(G))$ vertices with probability at least $1-1/n^4$, and
        \item has $\lambda(G')\ge \lambda(G)$ always, equality is achieved with probability at least $2\cdot 3^{-13}$.
    \end{enumerate}
\end{theorem}
Note that in the cut query model, the algorithm cannot contract an edge because the underlying cut function remains constant.
However, it can simulate the contraction of $e=(u,v)$ by merging the vertices $u$ and $v$ into a supervertex in all subsequent queries, which yields a contracted cut function.
We now define our \emph{$\tau$-star contraction} algorithm.
\begin{definition}[$\tau$-star contraction]
    Given an unweighted graph $G=(V,E)$ on $n$ vertices and a threshold $\tau\in [n]$, $\tau$-star contraction is the following operation.
    \begin{enumerate}
        \item Let $H=\{v \in V \mid d(v)\ge \tau\}$ be the set of high-degree vertices.
        \item Form the set $R\subseteq V$ by sampling each $v\in V$ independently with probability $p=\Theta(\log n/\tau)$.
        \item For each $u\in H\setminus R$, pick an edge from $E(u,R)$, if non-empty, uniformly at random and contract it (keeping parallel edges).
    \end{enumerate}
\end{definition}
Notice that $\tau$-star contraction depends on uniform edge sampling, since each vertex $u\in H\setminus R$ samples a neighbor in $R$.
In our case, we can use the uniform edge sampling procedure described above to sample edges uniformly from $E(u,R)$.

We now sketch the proof that $\tau$-star contraction yields a graph with few edges.
To bound the number of edges, partition the vertices remaining after the contraction into two sets, $R$, i.e. the center vertices sampled in the sampling process, and $L = V\setminus H$, i.e. the low degree vertices.
This analysis omits vertices in $H\setminus R$ that were not sampled into $R$ or contracted, because every $v\in H\setminus R$ has in expectation $d_G(v)\cdot \log n/ \tau \in \Omega(\log n)$ neighbors in $R$; hence with high probability it will have some neighbor in $R$ and be contracted.

To bound the number of edges observe that $|R|\le \tO(n/\tau)$ with high probability, and hence the number of edges between vertices in $R$ is $O(|R|^2)\le \tO(n^2/\tau^2)$.
In addition, since the degree of the vertices of $L$ is bounded by $\tau$, the total number of edges incident on $L$, even before the contraction, is $O(n\tau)$.
Hence, the total number of edges in $G'$ is with high probability $\tO(n^2/\tau^2 + n\tau)$.
The following lemma, proved in \Cref{sec:tau-contraction}, formalizes the guarantees of $\tau$-star contraction.
\begin{lemma}
    \label{lemma:tau-star-contraction}
    Let $G=(V,E)$ be an unweighted graph on $n$ vertices with a non-trivial minimum-cut $\lambda(G)$ and let $\tau\in [n]$.
    A $\tau$-star contraction of $G$ gives $G'$ that,
    \begin{enumerate}
        \item has at most $\tO(n^2/\tau^2+n\tau)$ edges with probability at least $1-1/n^4$, and
        \item has $\lambda(G')\ge \lambda(G)$ always, equality is achieved with probability at least $2\cdot 3^{-13}$.
    \end{enumerate}
\end{lemma}
We present in \Cref{algorithm:improved-min-cut} a procedure for finding the global minimum cut in two rounds using $\tau$-star contraction.
We will use it to prove \Cref{theorem:global-minimum-cut-2-rounds}, but first we need the following claim.
\begin{algorithm}[htbp]
    \caption{Unweighted Minimum Cut in $2$ Rounds}
    \label{algorithm:improved-min-cut}
    \begin{algorithmic}[1]
      \State \textbf{Input:} Graph $G=(V,E)$
      \State \textbf{Output:} A minimum cut $\cut \subseteq V$
      \State sample a set $R\subseteq V$ uniformly with probability $p=400\cdot \log n/\tau$
      \AdapRound{1}
        \State query $\mintcut_G(v)$ for all $v\in V$ \Comment{find $d_G(v)$ for all $v\in V$}
        \State for each $v\in V\setminus R$, let $n(v)\gets $ a random neighbor in $N_G(v)\cap R$ \label{lst:line:sample-edges} \Comment{use \Cref{lemma:edge-sampling-1-round}}
      \State $\tau \gets \max \{n^{1/3}, \delta(G)\}$
      \State $V'\gets V$
      \For{$v\in \{ V\setminus R \mid \mintcut_G(v) \ge \tau \}$}
        \State update $V'$ by contracting the edge $(v,n(v))$ \label{lst:line:tau-contraction}
      \EndFor
      \AdapRound{2}
          \State recover $G'=(V',E')$ with $\tO(n^{2-2/3}+n^{1+1/3})$ queries \Comment{use \Cref{corollary:graph-recovery-cut}, assuming $|E'|=\tO(n^{2-2/3}+n^{1+1/3})$}
      \State \Return $\min(\lambda(G'),\delta(G))$ \Comment{can return also a set $S\subseteq V$ achieving the cut}
    \end{algorithmic}
  \end{algorithm}
    
\begin{claim}
    \label{claim:implement-tau-star-contraction}
    With probability $1-n^{-3}$, line~\ref{lst:line:tau-contraction} of \Cref{algorithm:improved-min-cut} implements $\tau$-star contraction with $\tau=\max\set{n^{1/3},\delta(G)}$.
\end{claim}
\begin{proof}
    Notice that $R$ is sampled uniformly from $V$ and that line~\ref{lst:line:tau-contraction} is executed only for vertices in $H\setminus R$.
    It remains to show that the edge $(v,n(v))$ is sampled uniformly from $N_G(v)\cap R$, which indeed follows from \Cref{lemma:edge-sampling-1-round}.
    To analyze the success probability, note that the algorithm samples at most $n$ edges in line~\ref{lst:line:sample-edges} and by the guarantee of \Cref{lemma:edge-sampling-1-round} this fails with probability at most $n\cdot n^{-4}=n^{-3}$.
\end{proof}
\begin{proof}[Proof of \Cref{theorem:global-minimum-cut-2-rounds}]
    The main argument is that \Cref{algorithm:improved-min-cut} finds a minimum cut of $G$ with constant probability.
    Then, running this algorithm $O(\log n)$ times in parallel and returning the minimum output obtained gives the desired result.
    Notice that by \Cref{lemma:tau-star-contraction}, the connectivity of $G'$ can only increase, hence, if the algorithm finds a minimum cut of $G$ in some execution it will be returned.
    We assume throughout the proof that in all executions of the algorithm the size guarantee of \Cref{lemma:tau-star-contraction} holds, i.e. $|E'|=\tO(n^{2-2/3}+n^{1+1/3})$, note that this holds with probability $1-\tO(n^{-4})$.
    Using a union bound with \Cref{claim:implement-tau-star-contraction} above, we find that the algorithm succeeds with probability at least $1-n^{-2}$.

    By \Cref{claim:implement-tau-star-contraction} the graph $G'$ is the result of the $\tau$-star contraction with $\tau=n^{1/3}$, and thus $\lambda(G')=\lambda(G)$ with constant probability by \Cref{lemma:tau-star-contraction}.
    Therefore, recovering the graph $G'$ and returning the minimum of $\lambda(G')$ and $\delta(G)$ gives a minimum cut of $G$ with constant probability.

    To bound the query complexity, observe that the query complexity of the algorithm is dominated by the number of queries in the graph recovery step, which is given by $\tO(n^{2-2/3}+n^{1+1/3})=\tO(n^{4/3})$.
    In the first round, we use $\tO(n)$ queries as finding the degree of each vertex requires $n$ queries, and sampling a neighbor in $E(v,R)$ can be done using $O(\log^3 n/\log \log n)$ queries by \Cref{lemma:edge-sampling-1-round}.
    This concludes the analysis of the query complexity of the algorithm, which is $\tO(n^{4/3})$.

    Finally, notice that the algorithm proceeds in two rounds; 
    in the first round the algorithm only queries the degrees of the vertices and samples an edge in $E(v,R)$ for every $v$, which does not require any prior information on the graph.
    In the second round, it recovers $G'$ round using \Cref{theorem:graph-recovery}.
    This concludes the proof of \Cref{theorem:global-minimum-cut-2-rounds}.
\end{proof}

\subsection{\texorpdfstring{Unweighted Minimum Cut in $2r$ rounds} {Unweighted Minimum Cut in 2r rounds}}
This section provides an algorithm for finding a minimum cut in $2r$ rounds, thus proving \Cref{theorem:min-cut-k-rounds}.
The algorithm uses a known approach of first applying a contraction algorithm, yielding a graph with $\tO(n/\delta(G))$ vertices, and then packing $\delta(G)$ edge-disjoint forests in the contracted graph.
The main challenge in adapting this algorithm to low round complexity, is ensuring that the forests are edge disjoint when packing them in parallel.
This is achieved by leveraging a connection to cut sparsification, which shows that sampling edges with probabilities amplified by a factor of $k$ relative to those needed for sparsifier construction ensures finding $k$ edge disjoint forests within the sampled subset.
We note that we also offer another algorithm for finding a minimum cut, that has slightly worse query complexity of $O(rn^{1+1/r})$, and is based on constructing a cut sparsifier and the results of \cite{RSW18}, see \Cref{sec:min-cut-sparsifier-k-rounds}.

The contraction we use is the \emph{$2$-out contraction} procedure of \cite{GNT20}, which is obtained by sampling uniformly two edges incident to each vertex and contracting all the sampled edges.
The following theorem states the guarantees of the $2$-out contraction.
\begin{theorem}[Theorem 2.4 in \cite{GNT20}]
    \label{theorem:2-out-contraction}
    Let $G=(V,E)$ be an unweighted graph on $n$ vertices with a non-trivial minimum-cut value $\lambda(G)$.
    Then, applying a $2$-out contraction algorithm gives $G'$ where,
    \begin{enumerate}
        \item $G'$ has at most $\tO(n /\delta(G))$ vertices with high probability, and
        \item has $\lambda(G')\ge \lambda(G)$ always, equality is achieved with constant probability.
    \end{enumerate}
\end{theorem}
The other building block of the algorithm is maximal $k$-packing of forests.
\begin{definition}[Maximal $k$-Packing of Forests]
    Given a (possibly weighted) graph $G$, a maximal $k$-packing of forests is a set of edge-disjoint forests $T_1,\ldots,T_k$ that are maximal, i.e. for every $i\in [k]$ and $e\in E\setminus \cup_j T_j$, the edge set $T_i \cup \{e\}$ contains a cycle.
\end{definition}
\begin{lemma}
    \label{lemma:k-tree-packing}
    Given a (possibly weighted) graph $G$ on $n$ vertices and parameters $r\in\{1,\ldots,\log n\}$ and $k\in [n]$, one can find a maximal $k$-packing of forests of $G$ using $\tO(krn^{1+1/r})$ cut queries in $2r$ rounds.
    The algorithm is randomized and succeeds with probability $1-n^{-4}$.
\end{lemma}
We now outline the algorithm of \Cref{lemma:k-tree-packing}, for further details see \Cref{sec:k-tree-packing}.
The algorithm begins by creating a set of empty trees $T_1,\ldots,T_k$.
Then, in each of $r-1$ rounds it uniformly samples $\tO(kn^{1+1/r})$ edges from the graph, augments the trees $\set{T_i}$ by adding all sampled edges that do not form a cycle, and contracts all vertices that are connected in all trees.
We show that the number of edges decreases by an $n^{1/r}$ factor in each round; hence, after $r-1$ rounds the number of remaining edges is $O(n^{1+1/r})$.
Finally, the algorithm recovers all remaining edges by sampling $\tO(n^{1+1/r})$ edges uniformly, which recovers all edges by a standard coupon collector argument, and augments the trees appropriately.

To guarantee the decrease in the number of edges we leverage an elegant connection to cut sparsifier contraction.
The seminal work \cite{BK15} showed that if we sample (and reweight) every edge with probability proportional to its strong connectivity, then the resulting graph approximates all the cuts of the original graph with high probability.
\begin{definition}[Strength of an Edge]
    Let $G=(V,E,w)$ be a weighted graph.
    We say that $S\subseteq V$ is \emph{$\kappa$-strong} if the minimum cut of the induced graph $G[S]$ is at least $\kappa$.
    An edge $e\in E$ is $\kappa$-strong if it is contained within some $\kappa$-strong component $S$.
\end{definition}
We also need the following lemma concerning the existence of $\kappa$-strong components in a graph.
\begin{lemma}[Lemma 3.1 in \cite{BK96}]
    \label{lemma:edge-weight-connectivity}
    Every weighted graph $H$ with $n$ vertices and total edge weight at least $\kappa(n-1)$ must have a $\kappa$-strong component.
\end{lemma}
We can now detail the process of edge reduction of \Cref{lemma:k-tree-packing}.
Let $G$ be a graph with $m$ edges and $n$ vertices.
Observe that after one round of sampling, the algorithm contracts every $(m/n^{1+1/r})$-strong component of $G$, and thus by \Cref{lemma:edge-weight-connectivity} the number of edges remaining in the graph is at most $m/n^{1/r}$.
Let $C\subseteq V$ be a maximal $\Omega(m/n^{1+1/r})$-strong component of $G$, and assume the algorithm samples every edge of the graph with probability $p_e \approx k/\kappa$.
In expectation the algorithm samples at least $p_e \cdot m/n^{1+1/r} \approx k$ edges crossing each cut $S\subseteq C$, which also holds with high probability by standard concentration and cut counting results.
Therefore, for every $u,v\in C$, the minimum cut in the sampled subgraph has at least $k$ edges, hence they are connected in $k$ edge disjoint forests and will be contracted.
In conclusion, in every iteration the algorithm contracts all $(m/n^{1+1/r})$-strong components of the graph, and the number of edges in the graph decreases by a factor of $n^{1/r}$.

\begin{proof}[Proof of \Cref{theorem:min-cut-k-rounds}]
    The algorithm begins by performing a $2$-out contraction on $G$ to obtain a graph $G'$ with $\tO(n/\delta(G))$ vertices which preserves each non-trivial minimum cut of $G$ with constant probability by \Cref{theorem:2-out-contraction}.
    In parallel, the algorithm learns $\delta(G)$ by querying the degree of each vertex.
    Then, it packs $k=\delta(G)$ trees $\set{T_i}_{i=1}^k$ in $G'$ using \Cref{lemma:k-tree-packing}.
    Notice that since the minimum cut of $G$ is at most $\delta(G)$ and each edge has weight at least $1$ then $\delta(G)$ trees in $G'$ will include all edges in the minimum cut.
    Hence, returning the minimum between $\delta(G)$ and the minimum cut of $H=(V,E_H=\cup_i T_i)$ will yield the minimum cut of $G$ with constant probability.
    To amplify the success probability we repeat the algorithm $O(\log n)$ times in parallel and return the minimum cut found.

    To analyze the query complexity notice that by \Cref{lemma:edge-sampling-1-round}, sampling the edges for the $2$-out contraction takes $\tO(n)$ queries in one round.
    Hence, the algorithm constructs $G'$ in $\tO(n)$ queries in one round.
    Then, it packs $k=\delta(G)$ trees in $G'$, which has only $\tO(n/k)$ vertices, using $\tO(\delta(G)r(n/\delta(G))^{1+1/r})$ queries in $2r$ rounds by \Cref{lemma:k-tree-packing}.
    Hence, overall the algorithm require $\tO(\delta(G)^{-1/r}rn^{1+1/r})$ queries in $2r+1$ rounds.

    We conclude the proof by analyzing the success probability of the algorithm.
    Our implementation of the $2$-out contraction succeeds with probability $1-n\cdot n^{-4}$ since it only requires sampling $O(n)$ edges using \Cref{lemma:edge-sampling-1-round}.
    Then, each tree packing succeeds with probability $1-2n^{-2}$ by \Cref{lemma:k-tree-packing}.
    By a union bound, the entire algorithm succeeds with probability $1-n^{-1}$, which concludes the proof of \Cref{theorem:min-cut-k-rounds}.
\end{proof}

\subsection{\texorpdfstring{Weighted Minimum Cut in $O(r)$ rounds} {Weighted Minimum Cut in O(r) rounds}}
In this section we give an overview of our algorithm for minimum cut in a weighted graph with low adaptivity.
It is an adaptation of the minimum cut algorithm of \cite{MN20}, which reduces the problem to constructing a cut sparsifier and solving a monotone matrix problem which is easily parallelizable.
The main difficulty lies within the construction of a cut sparsifier;
while previous work in the cut-query model already provided cut sparsifiers \cite{RSW18, PRW24}, its core step, weight-proportional edge sampling, uses $O(\log n)$ rounds.
Our main technical contribution here is our two-round weight-proportional edge sampling primitive (\Cref{lemma:weighted-edge-sampling}), which enables the construction of a $(1\pm \epsilon)$-cut-sparsifier in $3r$ rounds and the completion of the entire reduction within $O(r)$ rounds.

The algorithm of \cite{MN20} for finding a minimum cut in a weighted graph is actually based on the $2$-respecting minimum-cut framework of \cite{Karger00};
where a spanning tree $T\subseteq E$ is called \emph{$2$-respecting} for a cut $\C \subseteq E$ if $|\C\cap T| \le 2$.
The main insight here is that a given a spanning tree $T$ that is $2$-respecting for a minimum cut $\C$, there are only $O(n^2)$ cuts that one needs to check to find a minimum cut, and they correspond to all pairs of edges in $T$.
Originally \cite{Karger00}, proposed a clever dynamic algorithm to calculate all these cuts quickly.
In contrast, the algorithm of \cite{MN20} leverages the structure of the $2$-respecting spanning tree to show that the number of potential cuts is actually much smaller; restricting attention to sub-trees, $\set{T_i}$, that are all paths and whose total size is bounded.
Furthermore, when $T$ is a path the task of finding a minimum cut reduces to the following problem.
\begin{definition}[Monotone Matrix Problem]
    \label{definition:monotone-cost-matrix}
    Let  $A\in \R_+^{a\times a}$ be a matrix.
    For every column $j$ denote the first row in which $j$ attains its minimum value by $j_s$ and  the last row in which $j$ attains its minimum value by $j_t$.
    The matrix is called \emph{monotone} if for every $i>j$ we have $j_t\le i_s$;
    i.e. the row in which the minimum value is found increases monotonically as we move to the right.

    The \emph{monotone matrix problem} is to find the minimum value of a monotone matrix (of dimension $a\times a$) while querying the value of as few entries as possible.
\end{definition}
In the context of the minimum cut problem, each entry in the matrix corresponds to a cut in $G$ and can be recovered using a single cut query.
We summarize some results from \cite{MN20} in the following lemma, whose proof is provided in \Cref{sec:2-resp-min-cut-reduction}.
\begin{restatable}{lemma}{mnsummary}
    \label{lemma:mn20-summary}
    Assume there is an algorithm to construct a $(1\pm\epsilon)$-cut-sparsifier for a graph $G$ on $n$ vertices using $q_S(n)$ queries in $r_S(n)$ rounds, and an algorithm that solves the monotone matrix problem on instances of size $n$ using $q_M(n)$ queries in $r_M(n)$ rounds.
    If in addition the function $q_M$ is convex, then it is possible to find a minimum cut of $G$ using $\tO(q_S(n) + q_M(n) + n)$ cut queries in $r_S(n)+r_M(n)+2$ rounds.
\end{restatable}
Therefore, in order to find a minimum cut it suffices to construct a cut sparsifier and then solve the monotone matrix problem.
Our main technical contribution in this section is a construction of a cut sparsifier with low adaptivity.
We begin by formally defining the notion of a cut sparsifier.
\begin{definition}
    Let $G=(V,E,w)$ be a weighted graph and let $H=(V,E_H,w_H)$ be a weighted subgraph of $G$.
    We say that $H$ is a \emph{$(1\pm \epsilon)$-cut-sparsifier} of $G$ if,
    \begin{equation*}
        \forall S\subseteq V,
        \qquad 
        \mintcut_H(S) \in (1\pm\epsilon)\mintcut_G(S)
        .
    \end{equation*}
\end{definition}
\begin{lemma}
    \label{lemma:sparsifier-weighted-k-rounds}
    Given a weighted graph $G=(V,E,w)$ on $n$ vertices with integer edge weights bounded by $W$, and a parameter $r\in \{1,2,\ldots,\log n\}$,
    one can construct a $(1+\epsilon)$-cut-sparsifier of $G$ using $\tO(\epsilon^{-2} rn^{1+(1+\log_n W)/r})$ cut queries and $3r+3$ rounds.
    The algorithm is randomized and succeeds with probability $1-n^{-2}$.
\end{lemma}
We now sketch the proof of \Cref{lemma:sparsifier-weighted-k-rounds}, for more details see \Cref{sec:sparsifier-construction}.
The sparsifier construction is similar to the one in \cite{RSW18,PRW24} but modified to work in $r$ rounds.
It follows the sampling framework of \cite{BK15}, that shows that if one samples every edge in a graph with probability proportional to its strong connectivity, then the sampled graph is a $(1\pm \epsilon)$-cut-sparsifier with high probability.
The algorithm is based on iteratively finding the strongest components of the graph using weight-proportional edge sampling, assigning a strength estimate to all their edges, and contracting them.
It repeats this process until all vertices are contracted.
Finally, it samples all the edges in the graph according to their strength estimates, and returns the sampled edges as the sparsifier.

Our algorithm improves on the round complexity of the existing constructions in two ways:
the first one is the adoption of a two-round weight-proportional sampling algorithm that eliminates an $O(\log n)$ factor in the round complexity.
The second improvement is in the number of contraction steps performed.
The algorithm of \cite{RSW18,PRW24} uses $O(\log (nW))$ steps, where in each step strengths are approximated within factor $2$.
Our algorithm uses instead $r$ steps, where each step contracts all components with strength within an $n^{(1+\log_n W)/r}$ factor.
Since the maximum component strength is $O(n^{1+\log_n W})$, the algorithm can contract all the components within $r$ steps.

Our second technical contribution is an algorithm that solves the monotone matrix problem in few rounds, the proof is provided in \Cref{sec:monotone-cost-matrix}.
\begin{lemma}
    \label{lemma:monotone-cost-matrix}
    There exists an algorithm that solves the monotone matrix problem on instances of size $n$ using $\tO(n^{1+1/r})$ queries in $r$ rounds.
\end{lemma}

The algorithm for the monotone matrix problem in \cite{MN20} uses a divide-and-conquer approach.
At each step the algorithm reads the entire middle column $j$ of the matrix $A$ and finds the minimum value in the column.
Then, it splits the matrix into two submatrices $A_L=A[1,\ldots, j_s; 1,\ldots,j]$ and $A_R=A[j_t,\ldots n; j+1,\ldots n]$ based on the location of the minimum value in the column.
The algorithm is guaranteed to find the minimum value in each column by the monotonicity property of the matrix.
Finally, it is easy to verify that as the size of the problem decreases by factor $2$ in each iteration, the overall complexity of the algorithm is $\tO(n)$.

Unfortunately, the algorithm of \cite{MN20} requires $O(\log n)$ rounds to complete all the recursion calls.
Our algorithm instead partitions the matrix into $n^{1/r}$ blocks and then recurses on each block, which requires a total of $O(n^{1+1/r})$ queries in $r$ rounds.

To conclude this section, notice that \Cref{theorem:weighted-min-cut} follows by combining \Cref{lemma:mn20-summary}, \Cref{lemma:sparsifier-weighted-k-rounds}, and \Cref{lemma:monotone-cost-matrix}.

\subsection{Open Questions}
\subparagraph*{Optimal Tradeoff Between Rounds vs Queries Complexity} 
An interesting open question is whether our constructive results have matching lower bounds or, alternatively, if better algorithms exist.
Currently, there are no known query lower bounds even for algorithms that use a single round of queries.
We note that connectivity can be solved in one round with $\tO(n)$ queries \cite{ACK21}, however it remains unclear whether this can be extended to finding a minimum cut, or whether there exists a lower bound of $\Omega(n^{1+\epsilon})$ cut queries for one round algorithms.

\subparagraph*{Adaptivity of Deterministic Algorithms}
A deterministic algorithm for finding a minimum cut in simple graphs using $\tO(n^{5/3})$ cut queries was recently shown in \cite{ASW25}, however its round complexity is polynomial in $n$, and it seems challenging to achieve polylogarithmic round complexity. 
In comparison, graph connectivity can be solved deterministically using $\tO(n^{1+1/r})$ cut queries in $O(r)$ rounds \cite{ACK21},\footnote{The result is for the weaker Bipartite Independent Set (BIS) query model, which can be simulated using cut queries.} though these techniques do not seem to extend to finding a minimum cut.

{\small
  \bibliographystyle{alphaurl}
  \bibliography{cut-query-adaptivity}
} %
\appendix

\section{Preliminaries}
\label{sec:preliminaries}
\subsection{Notation}
Let $G=(V,E)$ be some graph on $n$ vertices and $m$ edges.
For every $v\in V$ let $N_G(v)$ be the set of neighbors of $v$ in $G$ and let $d_G(v)=|N_G(v)|$ the degree of $v$ in $G$.
We also denote the minimum degree of $G$ by $\delta(G)=\min_{v\in V}d_G(v)$.
Additionally, denote the value of the minimum cut of $G$ by $\lambda(G)$.

For two disjoint subsets $S,T\subseteq V$ let $E(S,T)$ be the set of edges with one endpoint in $S$ and the other in $T$.
When $v$ is a single vertex we will use the shorthand $d_T(v) = |E(\{v\},T)|$.
In the case of weighted graphs, we denote the weight of a set of edges $E'\subseteq E$ by $w(E') \coloneqq \sum_{e\in E'} w(e)$.

\subsection{Cut Query Primitives}
Throughout the proof we will use the following cut query primitives to find the weight of edges between two sets of vertices, $w(E(S,T))$.
Notice that when the graph is unweighted we have $w(E(S,T))= |E(S,T)|$.
\begin{claim}
  \label{claim:s-t-num-edges}
  Let $S,T \subseteq V$ be two disjoint sets.
  It is possible to find $w(E(S,T))$ using $O(1)$ non-adaptive cut queries.
\end{claim}
\begin{proof}
  Observe that $\mintcut_G(S\cup T) = \mintcut_G(S)+\mintcut_G(T)-2w(E(S,T))$.
  Hence, we can find $w(E(S,T))$ using three queries, one for each of $\mintcut_G(S)$, $\mintcut_G(T)$, and $\mintcut_G(S\cup T)$.  
\end{proof}

Recovering an unknown graph using queries is a well studied problem.
In particular, the model of additive queries where given a set $S\subseteq V$ the query returns the total weight of edges in $S\times S$, $Q(S) \coloneqq \sum_{E \cap S\times S} w(e)$, has received a lot of attention \cite{CK08,BM11,BM12}.
Throughout the paper we shall use the following result from \cite{BM11}.
\begin{theorem}
  \label{theorem:graph-recovery}
  Let $G$ be a weighted graph with $m$ edges.
  There exists an algorithm that uses $\tO(m)$ non-adaptive additive queries to recover $G$.
\end{theorem}
The following claim shows how we can simulate additive queries using cut queries.
\begin{claim}
  \label{claim:additive-to-cut}
  Given a weighted graph $G$ on $n$ vertices, one can simulate $k$ additive queries using $n+k$ cut queries.
\end{claim}
\begin{proof}
  Begin by querying the singleton cuts, $\mintcut_G(\left\{ v \right\}) = \sum_{u\in V \mid (u,v)\in E} w(e)$ for all $v\in V$.
  Then for every $S\subseteq V$ notice that 
  \begin{equation*}
    Q(S)
    = \sum_{E \cap S\times S} w(e)
    = \sum_{v\in S} \mintcut_G(\{v\}) - \sum_{e\in S\times (V\setminus S)} w(e)
    ,
  \end{equation*}
  where the equality is by observing that the left-hand side is simply the sum of degrees of all vertices in $S$ minus $w(E(S,V\setminus S))$.
  The sum of degrees can be obtained by querying all $n$ singleton cuts once and $w(E(S,V\setminus S))$ can be obtained using $3$ cut queries by \Cref{claim:s-t-num-edges}.
  Therefore, any set of $k$ additive queries can be simulated using a fixed set of $n$ cut queries plus $k$ additional cut queries.
\end{proof}
Combining  \Cref{theorem:graph-recovery} with \Cref{claim:additive-to-cut} we immediately get the following corollary.
\graphrecoverycorollary*

Another useful primitive is the ability to estimate the degree of a vertex in a weighted graph.
The result is similar to Theorem 7 in \cite{Bshouty19} and is included for completeness.
\begin{lemma}
  \label{lemma:degree-estimation}
  Given a weighted graph $G$ on $n$ vertices, a designated vertex $v$, and a subset $S\subseteq V$ such that $v\not\in S$.
  There exists an algorithm that estimates $|E(v,S)|$ up to a $2^5$ multiplicative factor using $O(\log n)$ non-adaptive cut queries.
  The algorithm succeeds with probability $1-n^{-5}$.
\end{lemma}
Notice that setting $S=V\setminus \{v\}$ allows us to estimate the degree of $v$.
We now present the proof of the lemma.
\begin{proof}
  The algorithm is based on repeating a simple sampling procedure $O(\log n)$ times and returning the median estimate of the degree.
  Each sampling procedure works as follows.
  It begins by sampling subsets $S=S_0\supseteq S_1 \supseteq \ldots S_{2\log_{2} n}$ where each $S_i$ is obtained by keeping each vertex in $S_{i-1}$ independently with probability $1/2$.
  
  Denote, $j = \lceil \log_2 (|E(v,S_0)|) \rceil$ and let $a$ be the smallest $i$ such that $|E(v,S_{i})| =0$.
  We will show that with constant probability $a$ is close to $j$, the repeating the process $O(\log n)$ times and returning the median of the estimates will give us a good estimate of $|E(v,S)|$.
  To prove this, we will show that $\Probability{a \ge j+ 4 \rfloor} \le 1/8$ and $\Probability{a \le j- 4 \rfloor} \le 1/8$.
  This implies that $a$ is within $2^5$ of $j$ with probability at least $3/4$, and hence the median of $O(\log n)$ samples will be within $2^5$ of $j$ with probability at least $1-1/n^{\Omega(1)}$.

  We begin with showing that $\Probability{a \ge j+ 4 \rfloor} \le 1/4$.
  For every $e\in E$ let $X_{e,i}$ be the indicator random variable for the event that $e\in E(v,S_i)$ and let $X_i = \sum_{e\in E(v,S)} X_{e,i}$.
  Observe that, $\Exp{X_i}{} = n\cdot 2^{-i}$.
  Then, by the one-sided Chernoff bound (\Cref{theorem:chernoff-one-sided}) with $\mu=2^{-4}$ we have that,
  \begin{equation*}
    \Probability{|X_{j+4} - \Exp{X_{j+4}}{}| \ge 1}
    \le \exp\left( -\frac{2^8 \cdot 1  }{2+2^4} \right)
    \le\exp\left( -2^{3}\right)
    \le 1/8
    .
  \end{equation*}

  The proof for the other direction is similar, and is omitted for brevity.
  Hence, performing $O(\log n)$ instances of the algorithm and returning the median yields an estimate of $|E(v,S)|$ within a $2^5$ factor with probability at least $1-n^{-5}$.
\end{proof}

\subsection{Concentration Inequalities}
Throughout the proof we will use the following version of the Chernoff bound.
\begin{theorem}\label[theorem]{theorem:chernoff}
  Let $X_1,\ldots,X_m\in [0,a]$ be independent random variables.
    For any $\delta \in [0,1]$ and $\mu \geq \mathbb{E}\left[\sum_{i=1}^{m}X_i\right]$, we have
        \begin{align}
            \nonumber
            \mathbb{P}\left[
                \left|
                    \sum_{i=1}^{m}X_i - \mathbb{E}\left[\sum_{i=1}^{m}X_i\right]
                \right|
                \geq \delta \mu
            \right]
            \leq
            2\exp
            \left(
                -\frac{\delta^2\mu}{3a}
            \right)
            .
        \end{align}
\end{theorem}
We will also use the following one-sided version of the Chernoff bound.
\begin{theorem}\label[theorem]{theorem:chernoff-one-sided}
  Let $X_1,\ldots,X_m\in [0,a]$ be independent random variables.
  For $\mu \geq \mathbb{E}\left[\sum_{i=1}^{m}X_i\right]$, we have
      \begin{align}
        \forall \delta > 0,
        \qquad
          \nonumber
          \mathbb{P}\left[
              \sum_{i=1}^{m}X_i - \mathbb{E}\left[\sum_{i=1}^{m}X_i\right]
              \ge \delta \mu
          \right]
          \leq
          \exp
          \left(
              -\frac{\delta^2\mu}{(2+\delta) a}
          \right)
          .
      \end{align}
    \begin{align}
      \forall \delta\in [0,1],
      \qquad
      \nonumber
      \mathbb{P}
      \left[
          \sum_{i=1}^{m}X_i
          -         \mathbb{E}\left[\sum_{i=1}^{m}X_i
          \right]
          \leq -\delta \mu
      \right]
      \leq
      \exp
      \left(
          -\frac{\delta^2\mu}{2 a}
      \right)
      .
  \end{align}

\end{theorem}

We also use the following version of the Chernoff bound for negatively correlated random variables.
\begin{definition}
    Bernoulli random variables $B_1,\ldots,B_k \in \{0,1\}$ are \emph{negatively correlated} if 
  \begin{equation*}
    \forall I \subseteq [k],
    \forall j\in I,
    \forall b\in \set{0,1},
    \qquad
    \Probability{B_j = b \mid \bigwedge_{i\in I\setminus\set{j}} B_i = b} \le \Probability{B_j = b}
    ,
  \end{equation*}
  that is, conditioning on some variables taking the value $b$ makes it less likely that $B_j$ is $b$.
\end{definition}
The following lemma is a version of the Chernoff bound for negatively correlated Bernoulli random variables.
\begin{lemma}[\cite{DR98}]
  \label{lemma:negatively-correlated-chernoff}
  Let $a_1,\ldots,a_k\in[0,1]$ be constants and let $B_1,\ldots,B_k$ be negatively correlated Bernoulli random variables.
  Let $B\coloneqq \sum_{i=1}^k a_i B_i$ and $\mu\coloneqq\Exp{B}{}=\sum_{i=1}^{k} a_ip_i$.
  Then,
  \begin{equation*}
    \forall\epsilon\in(0,1),
    \qquad
    \Probability{B \not \in (1\pm \epsilon)\mu }
    \le 2\exp(-\epsilon^2\mu/3)
    .
  \end{equation*}
\end{lemma}

\section{Edge Sampling Primitives}
\label{sec:low-adaptivity}
In this section we provide two algorithms for sampling edges from a graph.
\subsection{Uniform Edge Sampling}
This section is devoted to proving \Cref{lemma:edge-sampling-1-round} by designing a procedure for uniform edge sampling from a possibly weighted graph.
As mentioned earlier, a similar algorithm was proposed in \cite{ACK21}, based on group testing, while our approach is based on graph recovery and slightly improves the query complexity of the algorithm by a factor of $O(1/\log \log n)$.
We will need the following bound for sparse recovery \cite{FSSZ23}.\footnote{A similar bound can be obtained by treating the edge set $E(s,T)$ as a bipartite graph on $|T|+1$ vertices and using the graph recovery result from \cite{BM11}.}
\begin{theorem}[Theorem 1.2 from \cite{FSSZ23}]
  \label{theorem:sparse-recovery}
  Let $M\in \set{-1,1}^{q\times a}$ be a random matrix where each entry is chosen independently and uniformly from $\set{-1,1}$, and let $s=(1-\delta)q$ and $\epsilon=\omega(\log \log a / \log a)$.
  If $a \le q^{1+1/(2-2\delta)-\epsilon}$, then with high probability every sparse $v\in \R^a$, such that $\|v\|_0\le s/2$ can be uniquely recovered from the measurement vector $Mv$.
\end{theorem}
\begin{corollary}
  \label{corollary:sparse-recovery} 
  There exists a matrix $M\in \{\pm 1\}^{q\times a}$ with $q=O(\log^2 a/\log\log a)$, such that every $(200 \log a)$-sparse vector $v\in\R^a$ can be recovered from the measurement vector $Mv$.
\end{corollary}
\begin{proof}
  Given $a$, let $s=200\log a$, and suppose $q,\delta$ satisfy $200\log a = (1-\delta)q$.
  This implies that 
  \begin{align*}
    q^{1+1/(2-2\delta)-\epsilon} 
    = 
    (200\ln a/(1-\delta))^{1+1/(2-2\delta)-\epsilon}
    \ge 
    \exp\Big(
      (1+1/(2-2\delta)-\epsilon)\cdot\ln\ln a 
      \Big)
    ,
  \end{align*}
  where the inequality is since $1-\delta\le 1$.
  To make sure the condition of \Cref{theorem:sparse-recovery} is satisfied it suffices to have that,
  \begin{align*}
    \ln a 
    \le 
    (1+1/(2-2\delta)-\epsilon)\ln\ln a 
    .
  \end{align*}
  Setting $\epsilon=1/2$ and solving for $\delta$ we find that requiring $\delta\ge 1-\log a/\log\log a$ satisfies the criteria.
  Therefore, $q=200\log a/(1-\delta)=O(\log^2 a/\log \log a)$ satisfies the conditions of \Cref{theorem:sparse-recovery}.
  Hence, with high probability a random matrix $M$ allows recovery of $v$ from the measurement vector $Mv$ and in particular such a matrix exists.
\end{proof}
\begin{proof}[Proof of \Cref{lemma:edge-sampling-1-round}]
  The algorithm works by running two processes in parallel.
  Both processes are based on successive refinements $T = T_0 \supseteq T_1 \ldots \supseteq T_{\log n}$, where each set $T_i$ is obtained by keeping each vertex in $T_{i-1}$ independently with probability $1/2$.
  The first process aims to find a level $j^*$ such that $1\le |E(s,T_{j^*})|\le 125\log n$.
  The second process attempts to recover all the edges of each level $E(s,T_i)$ using \Cref{corollary:sparse-recovery}, under the assumption that $1\le |E(s,T_i)|\le 125 \log n$.
  Combining the two, if the first process finds a level $j^*$ as desired, then the second process successfully recovers all the edges in $E(s,T_{j^*})$;
  the algorithm then selects uniformly at random one of the recovered edges and returns it.
  Notice that the query complexity of the algorithm is determined by performing $\log n$ instances of sparse recovery, one on each level $T_i$.
  If we were allowed an additional query round, the algorithm could perform the sparse recovery procedure only on the appropriate level $j^*$ which would reduce the query complexity by an $O(\log n)$ factor.

  We now detail the first process, i.e. the algorithm for finding a level $j^*$ for which $1\le|E(s,T_i)|\le125\log n$.
  The algorithm can obtain each $w(E(s,T_i))$ using $O(1)$ queries by \cref{claim:s-t-num-edges} and set $j^*$ to be the largest level such that $w(E(s,T_{j^*}))> 0$.
  If no such index exists then $E(s,T)$ is empty, and the algorithm can report this fact.
  We finish the analysis of the first process by showing that with high probability $1\le|E(s,T_{j^*})|\le 125\log n$.
  Note that it suffices to show  that there exists some level $j$ such that $25\log n \le |E(s,T_j)| \le 125\log n$ with high probability, which implies that $1\le|E(s,T_{j^*})| \le 125 \log n$ since $|E(s,T_j)|$ is monotonically non-increasing in $j$, and $j^*$ is the largest index such that $|E(s,T_{j^*})|>0$.

  We now prove that such a level $j$ exists with high probability.
  First, note that if $|E(s,T)|\le 125 \log n $ we are done.
  Otherwise, let $j = \lfloor \log (|E(s,T)|/(25\log n))\rfloor$.
  Let $X_e$ be the indicator random variable for the event that the edge $e\in E(s,T)$ is included in $E(s,T_j)$. 
  Notice that $\Pr[X_e] = 1/2^j$.  Then,
  \begin{align*}
    \Exp{|E(s,T_j)|}{}
    &=\Exp{\sum_{e\in E(s,T)}  X_e}{}
    = |E(s,T)|2^{-j}
    \in [25\log n, 50\log n]
    .
  \end{align*}
  Then, using the Chernoff bound,
  \begin{align*}
    \Probability{
      \left|
        \sum_{e\in E(s,T)}  X_e  - \Exp{\sum_{e\in E(s,T)}  X_e}{}
      \right| 
      > 25\log n}
    \le 2\exp(-50\log n/12)
    \le n^{-4}
    .
  \end{align*}
  Hence, with probability at least $1-n^{-4}$ there exists $j$ such that $25\log n \le |E(s,T_j)| \le 125\log n$.

  We now describe the procedure to recover all the edges in $E(s,T_{i})$ for some level $i$.
  If $|T_{i}|\le 200 \log n$ we can simply query separately every edge in $E(s,T_{i})$.
  Otherwise, Let $M\in \R^{|T_{i}|\times O(\log^2 n/\log\log n)}$ be the matrix guaranteed by \Cref{corollary:sparse-recovery} and let $\omega\in \R^{|T_{i}|}$ be the vector of edge weights in $E(s,T_{i})$.
  Assuming $|E(s,T_{i})|\le 200 \log n$, the vector $\omega$ is $(200\log n)$-sparse.
  Hence, $\omega$ can be uniquely recovered from the measurement vector $M\omega$ and we only need to compute the vector $M\omega$ using cut queries.
  For this application it is critical that the measurement matrix has entries in $\set{-1,1}$.
  For row $k$ of $M$, let $M_k^+ \subseteq T_{i}$ denote the indices corresponding to $+1$ entries and $M_k^- \subseteq T_{i}$ denote the indices corresponding to $-1$ entries.
  Hence, the inner product between the $k$-th row and $\omega$ is given by $(M\omega)_k = w(E(s,M_k^+))-w(E(s,M_k^-))$; which can be evaluated by finding each of $w(E(s,M_k^+)), w(E(s,M_k^-))$ using $O(1)$ cut queries by \Cref{claim:s-t-num-edges}.
  Therefore, the algorithm can calculate the measurement vector $M\omega$ using $O(\log^2 n/\log\log n)$ cut queries, find all the edges in $E(s,T_{i})$, and sample one uniformly at random.

  To conclude the proof we analyze the query complexity of the algorithm.
  As mentioned above, the first process can be performed using $O(\log n)$ queries.
  For the second process, note that the algorithm performs $O(\log n)$ instances of sparse recovery, one on each level $T_i$.
  Since each operation uses $O(\log^2 n/\log\log n)$ queries, the total number of cut queries required is $O(\log^3 n/\log\log n)$.
\end{proof}

\subsection{Weight-Proportional Edge Sampling}
In this section we prove \Cref{lemma:weighted-edge-sampling} and show how to sample an edge according to its weight using $O(\log^4 n )$ cut queries in $2$ rounds.
The proof uses the following lemma for estimating the weights of heavy edges.
\begin{restatable}{lemma}{hhhitterslemma}
  \label{lemma:heavy-hitters-non-adaptive}
  Given a graph $G$ on $n$ vertices with integer edge weights bounded by $W$, parameter $\alpha>0$, a source vertex $s\in V$, and a target vertex set $T\subseteq V$, one can find an approximate weight vector $\tilde{w}\in\R^{|T|}$ such that with probability $1-n^{-10}$,
  \begin{equation*}
    \forall e\in E(s,T),
    \qquad
    w(e) \le \tilde{w}(e) \le w(e) + w(E(s,T))/\alpha.
  \end{equation*}
  The algorithm uses $O(\alpha\log n)$ non-adaptive cut queries.
\end{restatable}
The proof of \Cref{lemma:heavy-hitters-non-adaptive} appears in \Cref{sec:heavy-hitters} and is based on the count-min data structure \cite{CM05} and standard analysis.
We are now ready to prove \Cref{lemma:weighted-edge-sampling}.
\begin{proof}[Proof of \Cref{lemma:weighted-edge-sampling}]
  We begin with an overview of the algorithm.
  In the first round, the algorithm subsamples the vertices in $T$ into sets $T=T_0\supseteq T_1\supseteq \ldots \supseteq T_{c \log n}$, and at each level $i$ it recovers all edges in $E(s,T_i)$ that have weight at least $w(E(s,T))/2^{i+5}$ using \Cref{lemma:heavy-hitters-non-adaptive}.
  In the second round, the algorithm finds the exact weights of all the recovered edges and keeps those with weight in $[w(E(s,T))/2^i,w(E(s,T))/2^{i+5}]$.
  It concludes by sampling an edge according the distribution $\{p_e\}$, where for every edge $e$ recovered in the $i$-th level $p_e \propto 2^i w(e)/w(E(s,T))$;
  with probability $1-\sum_e p_e$, the algorithm returns FAIL.
  The detailed algorithm is described in \Cref{alg:weighted-edge-sampling}.
    \begin{algorithm}[htbp]
        \caption{Weighted Edge Sampling}
        \label{alg:weighted-edge-sampling}
        \begin{algorithmic}[1]
        \State $T_0 \gets T$
        \For{$i=1$ to $\log (nW)$}
            \State $T_i \gets $ subsample $T_{i-1}$ by keeping each vertex independently with probability $1/2$
        \EndFor
        \AdapRound{1}
        \State $\tilde{w}_i \gets$ approximate weights of $E(s,T_i)$  \Comment{use \Cref{lemma:heavy-hitters-non-adaptive} with $\alpha=c\log (nW)$ for $c\ge 1$}
        \label{lst:line:weighted-sampling-approximate-weights}
        \State recover $w(E(s,T))$ \Comment{use \Cref{claim:s-t-num-edges}}
        \For{$i\in[\log nW]$}
            \State $E_i \gets \{e\in E(s,T_i) \mid \tilde{w}_i(e) \ge w(E(s,T))/2^{i+5}\}$
        \EndFor
        \If{$|\cup_i E_i| \ge 2^8\cdot\alpha$} \label{lst:line:fail-wegihted-sampling}
            \State \Return FAIL
        \EndIf
        \AdapRound{2}
        \State for every $e\in \cup_i E_i$ recover $w(e)$ using $O(1)$ queries
        \For{$i\in[\log nW]$}
            \State $F_i \gets \{  e\in E_i \mid w(E(s,T))/2^{i+5} \le w(e) < w(E(s,T))/2^{i} \}$
            \State $F_i \gets F_i \setminus \cup_{j<i} F_j$
            \For{$e\in F_i$}
                \State $p_e \gets 2^{i-5} w(e) / w(E(s,T))$
            \EndFor
        \EndFor
        \If{$\sum_{e\in \cup_i F_i} p_e >1 $} \label{lst:line:fail-probability-wegihted-sampling}
          \State \Return FAIL
        \EndIf
        \State \Return one edge picked from $\cup_i F_i$  according to probabilities $\{p_e\}$ and FAIL with the remaining probability $1-\sum_{e\in \cup_i F_i} p_e$ \label{lst:line:weighted-sampling-return-edge}
        \end{algorithmic}
    \end{algorithm}

    We assume throughout the proof that the weight approximation algorithm succeeds, and bound the probability of failure at the end.
    We also assume that the algorithm does not return FAIL in line~\ref{lst:line:fail-wegihted-sampling} or line~\ref{lst:line:fail-probability-wegihted-sampling}.
    The probability of these failures will be bounded at the end of the proof.

    We begin by showing that if the algorithm returns an edge $e$, then $e$ is returned with probability proportional to its weight.
    Fix some edge $e\in E(s,T)$ and let $j=\lfloor \log_2(w(E(s,T))/w(e)) \rfloor$.
    The algorithm returns the edge $e$ only if it is included in $F_j$ and subsequently picked in line~\ref{lst:line:weighted-sampling-return-edge}.
    By \cref{lemma:heavy-hitters-non-adaptive}, $\tilde{w}_j(e) \ge w(e)$ and hence if $e\in E(s,T_j)$, it will be included in $F_j$.
    In conclusion,
    \begin{align*}
        \Probability{\text{return $e$}}
        &= \Probability{e\in F_j} 
        \cdot \Probability{\text{$e$ picked in line~\ref{lst:line:weighted-sampling-return-edge}}} 
        = 2^{-{j}} \cdot \frac{w(e)\cdot 2^{j-5}}{w(E(s,T))}
        = \frac{w(e)}{2^5 w(E(s,T))}
        .
    \end{align*}
    Hence, $e$ is returned with probability proportional to its weight.
    Next, we show that the algorithm succeeds, i.e., returns some edge (and not FAIL), \emph{with constant probability}.
    Then, running $O(\log n)$ instances of the procedure in parallel and returning the edge returned by the first instance that does not return FAIL amplifies the probability of success to $1-n^{-5}$.
    Using a union bound with the event that the heavy-hitters algorithm fails, we conclude that the algorithm succeeds with probability $1-n^{-4}$.

    We first show that the algorithm passes the condition in line~\ref{lst:line:fail-wegihted-sampling} with constant probability.
    This follows by bounding the size of every set $E_i$.
    Notice that $E_i$ only includes edges of weight at least $w(E(s,T))/2^{i+5}$ since $\tilde{w}_i(e) \ge w(e)$ by \Cref{lemma:heavy-hitters-non-adaptive}.
    Therefore, $\Exp{|E_i|}{}\le 2^{i+5}\cdot 2^{-i} = 2^5$ where the inequality is since there are at most $2^{i+5}$ such edges in $E(s,T)$.
    Hence, by Markov's inequality, $\Probability{\cup_i |E_i| \ge 2^8 \log nW } \le 1/8$, and hence the algorithm passes the condition in line~\ref{lst:line:fail-wegihted-sampling} with at least this probability.

    We now show that if the algorithm passes the condition in line~\ref{lst:line:fail-wegihted-sampling}, then it returns some edge in line~\ref{lst:line:weighted-sampling-return-edge} with constant probability.
    For every $e\in E$ let $P_e$ be the random variable taking the sampling probability assigned to $e$ (set it to $0$ if $e\not\in \cup_i F_i$), and denote $j=\lceil \log_2(w(E(s,T))/w(e))\rceil$.
    Notice that $P_e=0$ with probability $1-2^{-j}$ and $P_e=2^{j-5}w(e)/w(E(s,T))$ with probability $2^{-j}$.
    Therefore, $\Exp{P_e}{} = w(e)/(2^5w(E(s,T)))$ and
    \begin{equation*}
        \Exp{\sum_{e\in E}P_e}{} 
        = \sum_{e\in E} \Exp{P_e}{}
        =2^{-5}
        .
    \end{equation*}
    We now given an upper bound on the value of $P_e$.
    Note that at the level $j=0$ we include all edges in $E(s,T)$ with weight at least $w(E(s,T))/2^5$ deterministically.
    For every other level, we include only edges $e$ with weight in $[w(E(s,T))/2^{j+5},w(E(s,T))/2^{j+4})$, since all the edges with larger weights that have been sampled at level $j$ are also included in some $F_i$ for $i<j$.
    Hence, for non deterministic $e$, we have $P_e \le 2^{-9}$.
    Therefore, we can apply a Chernoff bound to lower bound the tail of the sum.
    Specifically, let $X = \sum_{e \in E} P_e$.
    By the Chernoff bound (\Cref{theorem:chernoff}), for any $\delta \in (0,1)$,
    \[
        \Probability{|X-2^{-5}|\le 2^{-5}\cdot\delta} 
        \leq 2\exp\left(-\frac{\delta^22^{-5}}{3\cdot2^{-9}}\right)
        = \leq 2\exp\left(-\frac{4}{3}\right).
    \]
    Setting $\delta = 1/2$, we have that with constant probability, $X\in[2^{-6},2^{-5}+2^{-6}]$.
    By a union bound with the event that $|\cup_i E_i| \le \alpha$, we find that the algorithm passes all conditions and returns an edge with constant probability.

    To complete the proof we need to show that the algorithm uses $O(\log^2 n \log^2 (nW))$ queries in $2$ rounds.
    Observe that computing each approximate weight vector $\tilde{w}_i$ at line~\ref{lst:line:weighted-sampling-approximate-weights} requires $O(\alpha \log n)=O(\log n \cdot \log (nW))$ cut queries in one round by \Cref{lemma:heavy-hitters-non-adaptive}.
    Since this is repeated at $\log (nW)$ levels, the total number of queries in the first round is $O(\log n \cdot \log^2 n W)$.
    In the second round, the algorithm requires $O(\alpha)=O(\log (nW))$ cut queries to recover the exact weights since otherwise the algorithm would return FAIL in line~\ref{lst:line:fail-wegihted-sampling}.
    Finally, repeating the algorithm $O(\log n)$ times in parallel allows us to return an edge with high probability.
    In conclusion, overall the algorithm uses $O(\log^2 n \cdot \log^2 (nW))$ cut queries in two rounds.
\end{proof}

\subsection{Heavy Hitters}
\label{sec:heavy-hitters}
In this section we prove \Cref{lemma:heavy-hitters-non-adaptive}.
The proof is based on the count-min sketch data structure and standard analysis.
We begin by defining the count-min sketch data structure.
\begin{definition}[Count-Min Sketch]
    A Count-Min Sketch is a streaming data structure that maintains a 2 dimensional array of counters $C[1,\ldots,d][1,\ldots,w]$, where all the entries are initially zero.
    Given parameters $(\epsilon,\delta)$ we set $w=\lceil e/\epsilon \rceil$ and $d=\lceil \log(1/\delta) \rceil$.
    We also choose $d$ pairwise independent hash functions $h_1,\ldots,h_d:[n]\to [w]$.
\end{definition}
Whenever we insert an element into the data structure, the data structure increments the counter at the corresponding position by setting $C[i,h_i(x)]\gets C[i,h_i(x)]+1$ for all $i\in[d]$.
To estimate the frequency of an element $x$ we return the minimum over the counters $\min_{i\in d} C[i,h_i(x)]$.
The following lemma states the guarantees of the count-min sketch.
\begin{lemma}[Theorem 1 in \cite{CM05}]
    \label{lemma:count-min-sketch}
    Let $x$ be an element with frequency $f_x$ and estimated frequency $\tilde{f}_x$.
    The count-min sketch gives an estimate $\tilde{f}_x$ such that $f_x \le \tilde{f}_x \le f_x + \epsilon\cdot\sum_{y\in [n]} f_y$ with probability at least $1-\delta$.
\end{lemma}

We now prove \Cref{lemma:heavy-hitters-non-adaptive}.
We restate the lemma for ease of reference.
\hhhitterslemma*
\begin{proof}
    The algorithm is based on a count-min data structure with parameters $\epsilon=1/\alpha$ and $\delta=1/n^{10}$ to estimate the weights of the edges, where we treat the edge weight vector as the frequency vector we wish to estimate.
    Notice that there is an explicit mapping from the edges to the bins in the count-min sketch.
    Therefore, each counter in the count-min sketch is associated with the weight of some subset $F\subseteq E(s,T)$.

    Hence, we can recover all the weights with the requisite error and success probability using \Cref{lemma:count-min-sketch} if we recover the weight of $dw=O(\alpha\log n)$ such subsets.
    To conclude the proof, we show that each subset can be recovered using $O(1)$ queries.
    Notice that each subset $F\subseteq E(s,T)$ corresponds to a subset $T' \subseteq T$ such that $F=E(s,T')$.
    Therefore, we can recover the weight of $F$ by querying the weight of $E(s,T')$ using $O(1)$ queries by \Cref{claim:s-t-num-edges}.
\end{proof}

\section{\texorpdfstring{$\tau$-Star Contraction}{tau-Star Contraction}}
\label{sec:tau-contraction}
This section provides the guarantees of the $\tau$-star contraction, by proving \Cref{lemma:tau-star-contraction}.
The proof follows similar lines to the proof of the original star contraction \cite{AEGLMN22}, for which we now give a brief overview.
The proof follows by bounding the probability of contracting an edge of the minimum cut during the process.
Fixing a minimum cut $\C \subseteq E$ of $G$, denote by $c(v)$ the number of edges of $\C$ incident to $v$.
Observe that the probability that when a contraction of a vertex $v\in H\setminus R$ hits an edge in $\C$ is equal to $c_R(v)/d_R(v)$, where $c_R(v) = |\C \cap E(v,R)|$ and recalling $H=\set{v\in V \mid d(v)\ge \tau}$.
Therefore,
\begin{equation*}
    \Probability{\text{$\C$ is not contracted} \mid R}
    \ge 1-\prod_{v\in H\setminus R} \left(1-\frac{c_R(v)}{d_R(v)}\right)
    .
\end{equation*}
We bound this expression using the following result from \cite{AEGLMN22}.
\begin{proposition}[Proposition 4.1 in \cite{AEGLMN22}]
    \label{proposition:probability-optimization}
    Let $n$ be a positive integer, $a \in [0,1)$, and $b\ge 1$. 
    Define
\begin{align*}
    F(a,b) 
    = 
    \min_{x\in \R^n}
    \quad
     &\prod_{i\in [n]} (1-x_i)
    \\
     \text{subject to}
    \quad
    &\sum_{i\in [n]} x_i = b,
    \\
    & \max_{i\in [n]} 0\le x_i \le a.
\end{align*}
Then $F(a,b)\ge(1-a)^{\lceil b/a \rceil}$.
\end{proposition}
To apply this proposition we need to bound $\sum_{v\in H\setminus R} c_R(v)/d_R(v)$ and  $\max_{v\in H\setminus R} c_R(v)/d_R(v)$.
It can be shown that $c_R(v)/d_R(v)\approx c(v)/r(v)$, and hence for the proof sketch we will assume that $c_R(v)/d_R(v) = c(v)/d(v)$.
To bound the maximum, notice that $c(v)\le d(v)/2$ since otherwise one can move the vertex $v$ to the other side of the cut and reduce its size.
To bound the sum notice that,
\begin{equation}
    \label{eq:sum-property-star-contraction}
    \sum_{v\in V} \frac{c(v)}{d(v)}
    \le \sum_{v\in V} \frac{c(v)}{\delta(G)}
    \le \frac{2|\C|}{\delta(G)}
    \le 2
    , 
\end{equation}
where the last inequality is since $|\C|\le \delta(G)\le \tau$. 
We will also need the following useful proposition from \cite{AEGLMN22}.
\begin{claim}
    Let $G=(V,E)$ be a simple $n$-vertex graph and let $\C\subseteq E$ be the non-trivial minimum cut of $G$.
    Choose a set $R$ by putting each vertex of $V$ into $R$ independently with some probability $p$.
    Then, for any $v\in V$,
    \begin{equation*}
        \Exp{c_R(v)/d_R(v) \mid d_R(v) >0}{} = \frac{c(v)}{d(v)}
        .
    \end{equation*}
\end{claim}
Notice that the above claim also holds for the $\tau$-star contraction since we sample $R$ uniformly at random.
We now turn to prove \Cref{lemma:tau-star-contraction}.
\begin{proof}[Proof of \Cref{lemma:tau-star-contraction}]
    Throughout the proof we set $p=800\cdot \log n/\tau$.
    We begin by arguing the correctness of the algorithm.
    Since the algorithm only performs edge contractions, it is clear that it can only increase the edge connectivity of the graph.
    Therefore, it remains to show that it does not contract any edge of $\C$ with constant probability.
    We say that a set $R$ is \emph{good} if $\sum_{v\in V} \frac{c(v)}{d(v)} \le 8$ and $\max_{v\in V} \frac{c(v)}{d(v)} \le 2/3$.
    Notice that if $R$ is good then the probability that $\C$ is not contracted is bounded by,
    \begin{align*}
        \Probability{\text{$\C$ is not contracted}}
        &\ge \Probability{\text{$\C$ is not contracted} \mid R \text{is good}} \cdot \Probability{R \text{ is good}}
        \\
        &\ge \left( 1-\left( 1-\frac{2}{3} \right)^{\lceil 8/(2/3)\rceil} \right) \cdot \Probability{R \text{ is good}}
        = 3^{-12} \cdot \Probability{R \text{ is good}},
    \end{align*}
    where the first inequality is from the law of total probability and the second is from \Cref{proposition:probability-optimization}.
    Therefore, to prove the correctness guarantee of the lemma it suffices to show that $\Probability{R \text{ is good}} \ge 2/3$.

    We first bound the probability that $\sum_{v\in H\setminus R} c_R(v)/d_R(v) \ge 8$.
    Let $\delta(H) = \min\{ d_G(v) : v\in H\}$ be the minimum degree in $H$.
    Then notice that similarly to \eqref{eq:sum-property-star-contraction} we have that,
    \begin{equation*}
        \sum_{v\in H\setminus R} \frac{c(v)}{d(v)}
        \le \sum_{v\in H\setminus R} \frac{c(v)}{\delta(v)}
        \le \frac{2|\C|}{\delta(v)}
        \le 2
        ,     
    \end{equation*}
    where the last inequality is from $|\C|\le \delta(G)\le \delta(H)$.
    By the linearity of expectation we have that,
    \begin{equation*}
        \Exp{\sum_{v\in H\setminus R} \frac{c_R(v)}{d_R(v)}}{R}
        = \sum_{v\in H\setminus R} \Exp{\frac{c_R(v)}{d_R(v)}}{R}
        \le \sum_{v\in H\setminus R} \frac{c(v)}{d(v)}
        \le 2
        ,
    \end{equation*}
    where the first inequality is from $\Exp{c_R(v)/d_R(v) \mid d_R(v) >0}{} \ge \Exp{\frac{c_R(v)}{d_R(v)}}{R}$.
    Therefore, by Markov's inequality we have,
    \begin{equation*}
        \Probability{\sum_{v\in H\setminus R : c_R(v)>0} \frac{c_R(v)}{d_R(v)} \ge 8}
        \le \frac{1}{4}
        .
    \end{equation*}
    For the rest of the proof assume that this event does not occur.

    We now turn to bound the probability that $\max{v\in H\setminus R} c_R(v)/d_R(v) > 2/3$.
    Recall that we sample each vertex into $R$ with probability $p=800\cdot \log n/\tau$.
    Using the Chernoff bound, we obtain that $d_R(v)\ge 0.9 d_G(v)p$ with probability at least,
    \begin{equation*}
        \Probability{d_R(v) \ge 0.9 d_G(v)\log n /\tau} 
        \le 1-\exp(-(1/10)^2 d_G(v)800\log n /(2\tau)) 
        \le 1-1/n^4
        ,
    \end{equation*}
    where the last inequality is from $d_G(v)\ge \tau$.
    Therefore, $d_R(v)\ge 0.9 d_G(v)p$ for all $v\in H$ with probability at least $1-1/n^3$.
    Assume this event holds for the rest of the proof.
    
    Notice that similarly to $d_R(v)$, we have $\Exp{c_R(v)}{R} = c(v)p$.
    Since $c(v)/d_G(v)\le 1/2$ we have that $\Exp{c_R(v)}{R} \le pd_G(v)/2$.
    Hence, if $c_R(v)/d_R(v) > 2/3$ then $c_R(v) > 2d_R(v)/3\ge 2\cdot (9/10) d_G(v)p/3=(6/5)\Exp{c_R(v)}{R}$, where the last inequality is from $d_R(v) \ge 0.9 d_G(v)p$.
    Furthermore, since $d_R(v)\ge 0.9 d_G(v)p \ge 720 \log n$ we must have that $c_R(v) \ge 480 \log n$ for this event to occur.
    
    To bound the maximum, we split into two cases.
    The first, when $\Exp{c_R(v)}{R} \ge 240 \log n$.
    In this case by the Chernoff bound we have,
    \begin{equation*}
        \Probability{c_R(v) \ge (6/5)\cdot \Exp{c_R(v)}{R}}
        \le \exp(-0.2^2 \cdot 240 \log n / 3)
        \le n^{-3}
        .
    \end{equation*}
    The second case is when $\Exp{c_R(v)}{R} < 240 \log n$.
    In this case let $s =  240 \log n/\Exp{c_R(v)}{R} -1 \ge  240 \log n/\Exp{c_R(v)}{R} /2$.
    Notice that $s\ge 1$.
    Then, by the Chernoff bound we have that,
    \begin{equation*}
        \Probability{c_R(v) \ge 240\log n}
        \le \exp(-s \cdot \Exp{c_R(v)}{R} / 3)
        \le \exp(-240 \log n / 6)
        \le n^{-80}
        .
    \end{equation*}
    Combining the two cases we have that $\max_{v\in H\setminus R} c_R(v)/d_R(v)\le 2/3$ with probability at least $1-3/n^3$.
    Hence, we can conclude that $R$ is good with probability at least $2/3$.

    To conclude the proof we now show the size guarantee.
    Partition the vertices of $G'$ into $R$ and $U=V\setminus R$.
    We have seen above that for every $v\in H \setminus R$ we have that $d_R(v)\ge 0.9 d_G(v)p\ge 720 \log n$ with probability at least $1-1/n^4$.
    Therefore, every vertex in $H\setminus R$ is contracted with high probability.
    Hence, every remaining vertex $v\in U$ has degree at most $\tau$.
    
    To conclude the proof, notice that the total number of edges in the graph is bounded by the number of edges between vertices in $R$ plus the total number of edges incident to $U$.
    Since, $|R|\le n/\tau$ with high probability the number of edges between vertices in $R$ is at most $\tO(n^2/\tau^2)$.
    Furthermore, since the degree of the vertices of $U$ is bounded by $\tau$, the total number of edges incident on $U$ is at most $O(n\tau)$.
    Hence, the total number of edges in $G'$ is at most $O(n^2/2\tau^2+n\tau)$ with high probability.
\end{proof}

\section{Cut Sparsifier Construction}
\label{sec:sparsifier-construction}
In this section we provide an algorithm for finding a cut sparsifier with low adaptivity, proving \Cref{lemma:sparsifier-weighted-k-rounds}.
The algorithm is similar to the one in \cite{RSW18,PRW24}, where in the first stage the strength of edges is estimated by iteratively finding the strong components of the graph by subsampling, and then contracting the strong components.
This allows the algorithm to approximate the strength of every edge in the graph without sampling too many edges.
Finally, the algorithm finishes by sampling all the edges according to their strength estimates and returning the resulting graph as a cut sparsifier.
For the sampling step we use the following result, whose proof we provide in the end of the section and is based on \cite{PRW24}.
\begin{lemma}
    \label{lemma:prw-cut-sparsifier-reduction}
    Given $X=[(C_1,\beta_1),\ldots,(C_q,\beta_q)]$, where $\set{C_i}$ are a laminar family of vertex sets such that if $C_i \subset C_j$ then $i<j$, and $\beta_i$ is a strength estimate of $C_i$.
    Let $E(C_i)=E \cap (C\times C)$ and $F(C_i)= E(C_i)\setminus \cup_{j<i} E(C_j)$.
    If,
    \begin{enumerate}
        \item Every $e\in E$ is contained in some $C_i$.
        \item For every $e$ we have that $\kappa_e \ge \beta_i\ge \kappa_e/\alpha$.
    \end{enumerate}
    Then, there exists an algorithm uses $\tO(\epsilon^{-2} n \alpha)$ cut queries in $3$ rounds, and returns a quality $(1\pm \epsilon)$ cut sparsifier of $G$.
    The algorithm succeeds with probability $1-n^{-2}$.
\end{lemma}

The sparsifier construction of \cite{PRW24}, uses an approximation factor $\alpha=n^4$.
Unfortunately, approximating the strength within a factor of $4$ in each iteration requires $O(\log nW)$ steps.
This yields a round complexity of $O((\log n + \log W) \cdot \log n )$, as the edge sampling algorithm used in \cite{RSW18,PRW24} requires $O(\log n)$ rounds.
We note that in \cite{PRW24} the authors construct a weighted sparsifier in only $O(\log^2 n)$ rounds, by first approximating the strengths of the edges up to an $O(n^4)$ factor and then refining the strength estimate in the manner described above.
This enables the efficient construction of cut sparsifiers even in the presence of exponential weights.
Unfortunately, it is unclear whether this approach can be extended to the case of low adaptivity.

Our algorithm is able to run in fewer rounds by leveraging two ideas:
1) Utilizing our two-round weight-proportional edge sampling \Cref{lemma:weighted-edge-sampling}.
2) Applying coarser steps in the sampling process.
This increases the number of edges sampled in each step (and hence the query complexity) but reduces the number of steps required.

The procedure for a single step of strength estimation is presented in \Cref{algorithm:1-step-contraction}.
Note that our algorithm reduces weighted graphs to an unweighted multigraph by splitting each edge with weight $w$, into $w$ parallel edges and then running the algorithm.
The single step strength estimation algorithm is presented in \Cref{alg:cut-sparsifier-construction}.

\begin{algorithm}[htbp]
    \caption{One-Step Contraction}
    \begin{algorithmic}[1]
        \label{algorithm:1-step-contraction}
        \State \textbf{Input:} An unweighted multigraph $G=(V,E), K \ge 0$
        \State \textbf{Output:} A contraction of the vertex set $V'$ and strength estimates $F$
        \Procedure{One-Step-Contraction}{$G,\epsilon,K$}
            \State $\R^{|V|}\ni s\gets 0$
            \State $\rho\gets A \log n, \tau \gets \rho nK$ \Comment{$A$ is a sufficiently large constant}
            \AdapRound{1}
            \State query $\mintcut_G(v)$ for all $v\in V$ \Comment{find $d_G(v)$ for all $v\in V$}
            \For {$i \in [\tau]$}
                \State $w\gets$ sample $v\in V$ with probability $\mintcut_G(v)/\sum_{u\in V} \mintcut_G(u)$
                \State $s_{w} \gets s_{w}+1$
            \EndFor 
            \State $F\gets \emptyset$
            \AdapRound{2}
            \For{$v\in V$}
                \State $F \gets F \cup$ sample $s_v$ edges from $E(v,V'\setminus \{v\})$ in proportion to weight \Comment{use \Cref{lemma:weighted-edge-sampling}}
            \EndFor            
            \State $H\gets (V,F)$ \label{lst:line:reweight}
            \State recursively remove all cuts of $H$ with fewer than $(4A/5) \log n$ edges \label{lst:line:remove-small-cuts}
            \State $F \gets \{ (C,m/ (2n K)) \mid C \subseteq V, \text{ $C$ is a connected component of $H$} \}$
            \State form $V'$ by contracting all connected components of $H$
            \State \Return $V',F$
        \EndProcedure
    \end{algorithmic}
\end{algorithm}

\begin{algorithm}[htbp]
  \caption{Find Strong Connected Components}
  \label{alg:cut-sparsifier-construction}
  \begin{algorithmic}[1]
    \State \textbf{Input:} Unweighted multigraph $G=(V,E)$, maximum number of parallel edges $W$.
    \State \textbf{Output:} Strong components and strength estimates $(C_1,\beta_1),\ldots,(C_q,\beta_q)$.
    \State $H \gets (V,\emptyset)$
    \State $\gamma \gets 1+\log_n W$
    \State $X\gets \emptyset$
    \For{$j \in [r]$}
        \State $V',F \gets $ONE-STEP-CONTRACTION($G,\epsilon,2n^{\gamma /r}$)
        \State $X\gets X\cup F$
        \State update $G$ by contracting all the vertices $V'$
    \EndFor
    \State \Return $F$
  \end{algorithmic}
\end{algorithm}

The main technical claim of this section is the following, whose proof is deferred to later in the section.
\begin{claim}
    \label{claim:strength-estimation-guarantee}
    \Cref{alg:cut-sparsifier-construction} returns a list of strength estimates $(C_1,\beta_1),\ldots,(C_q,\beta_q)$ such that 
    1) for every edge $e\in C_i$ we have, $\kappa_e \ge \beta_i\ge \kappa_e/n^{\gamma/r}$.
    2) every edge $e\in E$ is contained in some $C_i$.
\end{claim}
Using this we can now prove the \Cref{lemma:sparsifier-weighted-k-rounds}.
\begin{proof}[Proof of (\Cref{lemma:sparsifier-weighted-k-rounds})]
    Correctness of the algorithm follows from combining \Cref{lemma:prw-cut-sparsifier-reduction} and \Cref{claim:strength-estimation-guarantee}.

    To analyze the query complexity and round complexity, we begin by analyzing ONE-STEP-CONTRACTION.
    In the first step of the sampling the algorithm recovers the degree of each vertex, requiring $O(n)$ queries in a single round.
    Then, sampling using \Cref{lemma:weighted-edge-sampling} requires $O(\log^2 n \log^2 (nW))$ queries per edge, all of which run in $2$ parallel rounds.
    Hence, executing ONE-STEP-CONTRACTION once requires $O(\rho nK)=\tO( n^{1+\gamma/r})$ queries in $3$ rounds.
    Recalling that we execute the procedure in $r$ iterations the overall query complexity of strength estimation is $\tO(r \cdot \rho n^{1+\gamma/r})$ in $3r$ rounds.
    Then, we require an additional $\tO(\epsilon^{-2} n^{1+\gamma/r})$ queries in $3$ rounds to perform \Cref{lemma:prw-cut-sparsifier-reduction}.
    In conclusion, the sparsifier construction requires $\tO(\max\{\epsilon^{-2},r\}\cdot n^{1+\gamma/r})$ queries in $3r+3$ rounds.
\end{proof}

\subsection{\texorpdfstring{Proof of \Cref{claim:strength-estimation-guarantee}} {Proof of Main Claim}}
We begin by presenting two useful claims about the contraction procedure.
The proof of the claims is deferred to the end of the section.
\begin{claim}
    \label{claim:weak-connected-sparsifier}
    Let $G=(V,E)$ be an unweighted multigraph, fix some $K>0$, and let $F$ be the result of applying \Cref{algorithm:1-step-contraction} on $G$ with parameter $K$.
    Then, with probability at least $1-n^{-4}$, $F$ contains no $m/(2nK)$ weak components of $G$.
\end{claim}
\begin{claim}
    \label{claim:strong-connected-contraction}
    Let $G=(V,E)$ be an unweighted multigraph, fix some $K>0$, and let $F$ be the result of applying \Cref{algorithm:1-step-contraction} on $G$ with parameter $K$.
    Then, with probability at least $1-n^{-4}$, every $m/(nK)$-strong component of $G$ is in $F$.
\end{claim}
\begin{proof}[Proof of \Cref{claim:strength-estimation-guarantee}]
    We begin by assuming that the events described in \Cref{claim:weak-connected-sparsifier} and \Cref{claim:strong-connected-contraction} hold for all $r$ levels, this happens with probability at least $1-2r/n^{-4}$.
    We will now show that the strength estimates in the $i$-th level of the algorithm are correct.
    Denote the number of edges (including parallel edges) in the graph at the $i$-th stage of the algorithm by $m_i$.
    By \Cref{lemma:edge-weight-connectivity} the maximum strength of any edge in the graph at the $i$-th stage is at most $m_i/n$.
    Furthermore, by \Cref{claim:weak-connected-sparsifier} no $m_i/(2nK)$ weak edges are included in $F$ at the $i$-th stage.
    Since we report the strength of the edges in $F$ at the $i$-th stage as $\beta_i = m_i/(2n^{1+i\gamma/r})$, we have that we approximate the strength of every edge in $F$ by factor of at most $2n^{\gamma/r}$.

    It remains to show that every edge $e\in E$ is contained in some $C_i$.
    By \Cref{claim:strong-connected-contraction} every edge of strength $m_i/(2n^{1+\gamma/r})$ is contracted in the $i$-th step of the algorithm.
    Therefore, the algorithm reduces the maximum edge strength in the graph by factor $2n^{\gamma/r}$ in each iteration.
    After $r$ iterations the maximum edge strength in the graph is at most $m_0/(2^r n^{\gamma})\le 2^{-r}$.
    Since the minimal strength of any edge in the graph is at least $1$, we have that the algorithm contracts all edges in the graph after $r$ iterations.
\end{proof}
To conclude we provide the proofs of \Cref{claim:strong-connected-contraction} and \Cref{claim:weak-connected-sparsifier}.
\subsubsection{Proofs of Technical Claims}
Throughout this section we will use the following lemma from \cite{Karger93} which bounds the number of cuts in a graph.
\begin{lemma}
    \label{lemma:cut-counting}
    Let $G=(V,E)$ be a graph on $n$ vertices with minimum cut value $\lambda(G)$.
    For every $\alpha>1$ the number of cuts of value at most $\alpha\lambda(G)$ is at most $O(n^{2\alpha})$.
\end{lemma}

In this section we provide the proofs for \Cref{claim:strong-connected-contraction} and \Cref{claim:weak-connected-sparsifier}.
This concludes all the proofs required for the main results of this section.
We begin by proving \Cref{claim:strong-connected-contraction}.
\begin{proof}
    Let $C$ be some $(m/nK)$ strong component of $G$, and fix some cut $S\subseteq C$.
    For every $e\in E$ let $X_e$ be the indicator for the event that $e$ was sampled into $H$.
    Notice that $\{X_e\}_{e\in E(C)}$ are negatively correlated Bernoulli random variables.
    Using the above we find that 
    \begin{equation*}
        \Exp{|E_H(S,V\setminus S)|}{}
        = \sum_{e\in E(S,V\setminus S)} \Exp{X_e}{}
        = p_e\cdot |E(S,V\setminus S)|
        .
    \end{equation*}
    We now bound $p_e$.
    Notice that every edge in $C\times C$ is sampled into $H$ with probability $p_e = 1-(1-1/m)^{\tau}$.
    Hence,
    \begin{align*}
        p_e 
        &= 1-(1-1/m)^{\rho nK}
        = 1-(1-\rho nK/m\cdot(\rho nK)^{-1})^{\rho nK}
        \le \frac{\rho n K}{m}
        ,
    \end{align*}
    where the first inequality is from $(1+x/n)^n \le \exp(x)$.
    Similarly, we can lower bound $p_e$ using,
    \begin{align*}
        p_e 
        &= 1-(1-1/m)^{\rho nK}
        \ge 1-\exp\left( \frac{\rho n K}{m} \right)
        = \frac{\rho n K}{m} - O\left( \left( \frac{\rho n K}{m} \right)^2 \right)
        ,
    \end{align*}
    where the last inequality is from $(1-x/n)^n \le \exp(x)$ and the last equality is from the Taylor expansion of $\exp(x)$.
    Notice that the second element of the above is negligible whenever the algorithm does not sample all the remaining edges, in which case we are done.
    Therefore,
    \begin{equation*}
        \Exp{|E_H(S,V\setminus S)|}{}
        \ge \frac{m}{nK} \frac{\rho n K}{m}
        = \rho
        = A \log n
        .
    \end{equation*}
    Denoting $|E(S,V\setminus S)| = \alpha \lambda(G[C])$ and applying the Chernoff bound for negatively correlated Bernoulli random variables (\Cref{lemma:negatively-correlated-chernoff}) we obtain,
    \begin{align*}
    &\Probability{
        \big| |E_H(S,V\setminus S)|-  \Exp{E_H|(S,V\setminus S)|}{} \big| 
        \le (1/10)\cdot \frac{\rho n K}{m} |E(S,V\setminus S)|
    }
    \\
    &\le 2\exp\left( 
        -\frac{1}{300}
        \cdot \frac{\rho n K}{m} |E(S,V\setminus S)|
        \right)
    \\
    &= 2\exp\left( 
    -\frac{1}{300}
    \cdot \frac{A nK \log n}{m}
    \alpha \lambda(G)
    \right)
    \le
    2\exp\left( 
    -\frac{A \alpha \log n}{300}
    \right)
    ,
    \end{align*}
    where the last inequality is from $\lambda(G[C])\ge  m /(nK)$, by our assumption on $C$.
    Notice that if this event occurs then the cut $S$ has at least $(9/10)\cdot A\epsilon^{-2} \log n$ edges and is therefore not removed in line~\ref{lst:line:remove-small-cuts}.
    We now show that this indeed occurs with high probability for all cuts $S\subseteq V$.
    To do so, we will use the cut counting lemma (\Cref{lemma:cut-counting}).
    Using the lemma we will apply a union bound over all the cuts of $G[C]$.
    \begin{align*}
    &\Probability{\exists S\subseteq C: 
    \big| |E_H(S,V\setminus S)|-  \Exp{E_H|(S,V\setminus S)|}{} \big| 
    \ge 1/10 \cdot \frac{\rho n K}{m} |E(S,V\setminus S)|}
    \\
    &\le
    \sum_{\alpha=1}^{\infty} 
    O(n^{2\alpha})\cdot 2\exp\left( 
        -\frac{A \alpha \log n}{300}
        \right)
        =
        \sum_{\alpha=1}^{\infty} 
        O(n^{\alpha(2-A/300)})
        .
    \end{align*}
    Hence, for every large enough $A$ we obtain that the probability can be bounded by $O(1/n^5)$.
    Using a union bound over the $O(n)$ $m/(nK)$-strong components of $G$ we find that every $m/(nK)$-strong component of $G$ is connected in $H$ with probability at least $1-O(n^{-4})$.
    Therefore, it will be added to the list $F$.
    This concludes the proof.
\end{proof}

We now turn to prove \Cref{claim:weak-connected-sparsifier}.
\begin{proof}
Let $e$ be an edge of $G$ that is $ m/(2nK)$-weak.
Notice that this implies that it is not contained in any $ m/(2nK)$-strong component of $G$.
Therefore, if we show that all components that are $ m/(2nK)$-weak are removed in line~\ref{lst:line:remove-small-cuts} then we are done.
Begin by observing that removing at most $n$ cuts in line~\ref{lst:line:remove-small-cuts} of \Cref{algorithm:1-step-contraction} suffices to remove all the $ m/(2nK)$-weak components.

Let $\{S_1,S_2,\ldots, S_r\}$ be a set of minimum size cuts that suffices to remove all the $ m/(2nK)$-weak components.
Since they are minimal cuts in $ m/(2nK)$ weak components we have that $\mintcut_G(S_i) \le  m/(2nK)$.
To conclude the proof we show that $|E_H(S,V\setminus S)| \le (4c\epsilon^{-2} /5) \log n$.
Now fix some cut $S_i$, similarly to before we have
\begin{equation*}
    \Exp{E_H|(S_i,V\setminus S_i)|}{}
    = \sum_{e\in E(S_i,V\setminus S_i)} \Exp{X_e}{}
    = p_e\cdot |E(S_i,V\setminus S_i)|
    .
\end{equation*}
Notice that the worst case is when $|E(S_i,V\setminus S_i)| =  m/(2nK)$.
Using the same bound on $p_e$ as above and the Chernoff bound for negatively correlated Bernoulli random variables (\Cref{lemma:negatively-correlated-chernoff}) we obtain,
\begin{align*}
    &\Probability{\big| |E_H(S_i,V\setminus S_i)|-  \Exp{|E_H(S_i,V\setminus S_i)|}{} \big| 
    \ge \frac{1}{10}\cdot \frac{\rho n K}{m} |E(S_i,V\setminus S_i)|}
    \\
    &\le 2\exp\left( 
        -\frac{1}{300}
        \cdot \frac{A\log n}{\sqrt{n}}
        |E(S_i,V\setminus S_i)|
        \right)
    \le 2\exp\left( 
    -\frac{A\log n}{600}
    \right)
    ,
\end{align*}
where the last inequality is from using that the worst case is $|E(S_i,V\setminus S_i)| =  m/(2nK)$.
Applying the union bound over all $r\le n$ minimum cuts $S_1,\ldots,S_r$ we obtain that all the cuts are removed with probability at least $1-1/n^4$ for a large enough $A$.
Hence, no $ m/(2nK)$-weak components remain in $H$ and the claim is proven.
\end{proof}

\subsection{Proof of \cite{PRW24} Reduction (Lemma \ref{lemma:prw-cut-sparsifier-reduction})}
In this section we provide the proof of \Cref{lemma:prw-cut-sparsifier-reduction}.
The algorithm for constructing cut sparsifiers from \cite{PRW24} is presented in \Cref{alg:prw-sparsifier}.
The following lemma states the guarantees of the algorithm.
\begin{lemma}[Lemma 4.8 from \cite{PRW24}]
    Given $X=[(C_1,\beta_1),\ldots,(C_q,\beta_q)]$, where $C_i$ are a laminar family of vertex sets such that if $C_i \subset C_j$ then $i<j$ and $\beta_i$ is a strength estimate of $C_i$.
    Let $E(C_i)=E \cap (C\times C)$ and $F(C_i)= E(C_i)\setminus \cup_{j<i} E(C_j)$.
    If,
    \begin{enumerate}
        \item Every $e\in E$ is contained in some $C_i$.
        \item For every $e$ we have that $\kappa_e \ge \beta_i\ge \kappa_e/\alpha$.
    \end{enumerate}
    Then, with probability $1-n^{-5}$ \Cref{alg:prw-sparsifier} returns a quality $(1\pm\epsilon)$-cut sparsifier of $G$.
\end{lemma}

\begin{algorithm}[htbp]
    \caption{\cite{PRW24} Sparsifier Construction}
    \begin{algorithmic}[1]
        \label{alg:prw-sparsifier}
        \State \textbf{Input:} An weighted $G$, $X=[(C_1,\beta_1),\ldots,(C_q,\beta_q)]$, where $C_i$ are a laminar family of vertex sets such that if $C_i \subset C_j$ then $i<j$ and $\beta_i$ is a strength estimate of $C_i$, parameter $\epsilon$.
        \State \textbf{Output:} Cut sparsifier of $G$
        \State $G'\gets G, H\gets (V,\emptyset)$
        \For{$i \in [q]$}
            \State $\mu_i \gets \frac{\epsilon^{-2} c_1 \log^2 n}{\beta_i} w(F(C_i))$ \Comment{where $w(F(C_i))$ is the total weight of edges in $F(C_i)$, $c_1$ is a sufficiently large constant}
            \State sample $\mu_i$ edges from $F(C_i)$ in proportion to weight, give each edge that was sampled at least once weight $w(e)/p_e$ where $p_e = (1-(1-w(e)/W(C_i))^{\mu_i})$ \Comment{use \Cref{lemma:weighted-edge-sampling}}
            \State Contract $C_i$ in $G'$ and add the sampled edges to $H$
        \EndFor
        \State \Return $H$
    \end{algorithmic}
\end{algorithm}
We will also need the following result from \cite{PRW24} which states that \Cref{alg:prw-sparsifier} performs few edge sampling procedures.
\begin{claim}[Lemma G.43 from \cite{PRW24}]
    \label{claim:prw-sparsifier-sampling-complexity}
    The algorithm \Cref{alg:prw-sparsifier} calls $\tO(\epsilon^{-2} n \alpha)$ instances of weight-proportional sampling in parallel.
\end{claim}
We now proceed to proving \Cref{lemma:prw-cut-sparsifier-reduction}.
\begin{proof}
    Begin by observing that we can run all iterations of the main loop of \Cref{alg:prw-sparsifier} in parallel.
    This is because it can simulate the contraction of $C_i$ in the graph $G'$.
    Note that $q\le n$  since we can only contract $n$ times before all vertices are contracted.
    Then, the algorithm needs to compute $W(F(C_i))$ for all $C_i$, notice that this is equal to,
    \begin{equation*}
        w(F(C_i))
        = \sum_{(u,v) \in C_i \times C_i} w(u,v) 
        - \sum_{j<i: C_j \subseteq C_i} w(F(C_j))
        .
    \end{equation*}
    Note that this is a linear combination of $O(n)$ additive queries, therefore by \Cref{claim:additive-to-cut} we can compute this in $O(n)$ cut queries in one round.

    We now analyze the number of queries required to run \Cref{alg:prw-sparsifier}.
    Notice that the algorithm requires $\tO(\epsilon^{-2} n \alpha)$ instances of weight-proportional sampling in parallel (\Cref{claim:prw-sparsifier-sampling-complexity}).
    Therefore, the overall query complexity of the algorithm is $\tO(\epsilon^{-2} n \alpha)$ in $3$ rounds.
    
    We conclude by analyzing the success probability of the algorithm.
    By a union bound all weight-proportional sampling procedures succeed with probability at least $1-\tO(\epsilon^{-2} n^{-3} \alpha)$.
    Assuming this, then we have by \Cref{claim:prw-sparsifier-sampling-complexity} that the algorithm succeeds with probability at least $1-n^{-5}$.
    Applying a union bound again we find that the algorithm succeeds with probability at least $1-n^{-1}$, by noting that $\alpha,\epsilon^{-2}=o(n)$ whenever we use $o(n^2)$ cut queries.
\end{proof}

\section{\texorpdfstring{Maximal $k$-Packing of Forests} {Maximal k-Packing of Forests}}
\label{sec:k-tree-packing}
In this section we provide a low adaptivity algorithm for maximal $k$-packing of forests.
The algorithm begins by instantiating $k$ empty trees $T_1,\ldots,T_k$.
Then, it samples $O(kn^{1+1/r}\log n)$ edges uniformly from the graph and adds them to the trees if they do not form a cycle.
After that, it contracts all vertices that are connected in all trees.
The proof follows by showing that the number of edges in the graph is reduced by factor $n^{1/r}$ in each iteration.
Then, after $r-1$ iterations the graph has at most $n^{1+1/r}$ edges.
Finally, the algorithm recovers all remaining edges and augments the trees with all edges that do not form a cycle.
The full algorithm is presented in \Cref{alg:k-tree-packing}.
\begin{algorithm}[htbp]
    \caption{Maximal $k$-Packing of Forests Construction}
    \label{alg:k-tree-packing}
    \begin{algorithmic}[1]
        \State \textbf{Input:} Graph $G=(V,E,w)$, parameters $r,k$.
        \State \textbf{Output:} $k$ edge disjoint maximal forests.
        \State $H \gets (V,\emptyset)$
        \State $T_i \gets \emptyset$ for $i\in [k]$
        \State $\tau \gets A k n^{1+1/r} \log n$ \Comment{for constant $A>0$}
        \State $U \gets V$
        \For{$j \in [r]$}
        \AdapRound{}
            \State $\forall v\in V, \tilde{d}_j(v)\gets$ estimate $E(v,V\setminus U(v))$ up to $2^5$ factor \Comment{use \Cref{lemma:degree-estimation}, $U(v)$ are the vertices in the same supervertex as $v$}
            \State $F \gets $ sample $\tau$ edges uniformly from $G$ \label{lst:line:k-forest-sampling}
        \For{$e\in F, i\in [k]$}
            \If{$T_i\cup e$ is cycle free}
                \State $T_i \gets T_i \cup e$
                \State \textbf{continue} to the next edge
            \EndIf
        \EndFor
        \For{$\{u,v\} \in \binom{V}{2}$}
            \If{$u,v$ are connected in every forest $T_1,\ldots,T_k$} \label{lst:line:contract-tree}
                \State update $U$ by contracting $u,v$
            \EndIf
        \EndFor
        \EndFor
        \State \Return $\{T_i\}_{i=1}^k$
    \end{algorithmic}
\end{algorithm}

\begin{proof}[Proof of \Cref{lemma:k-tree-packing}]
    In the proof we shows that \Cref{alg:k-tree-packing} returns a maximal $k$-packing of forests.
    Throughout the proof we assume that the graph $G$ is unweighted, this is possible since we will use uniform edge sampling which us the same for weighted and unweighted graphs.
    Denote the number of edges in the $j$-th iteration of the main loop by $m_j$, we begin by proving that after $r-1$ iterations of the main loop, the graph $G$ has at most $n^{1+1/r}$ edges.
    For now, fix some iteration $j$.

    We begin by detailing how the uniform sampling in line~\ref{lst:line:k-forest-sampling} is performed.
    The algorithm first samples a vertex with probability $p_v = \tilde{d}_j(v)/\sum_{w\in V} \tilde{d}_j(w)$, and then samples one of the edges in $E(v,V\setminus U(v))$.
    Observe that the algorithm only samples edges outside of the supervertices, which are exactly the edges of the contracted graph.
    Therefore, the probability that the algorithm samples an edge $e=(u,v)$ in a single sample is given by,
    \begin{equation*}
        \pi_e 
        = \frac{\tilde{d}_j(v)}{\sum_{w\in V} \tilde{d}_j(w)}\cdot \frac{1}{d_G(v)} 
        + \frac{\tilde{d}_j(u)}{\sum_{w\in V} \tilde{d}_j(w)}\cdot \frac{1}{d_G(u)}
        \ge \frac{1}{2^{10} m_j}
        ,
    \end{equation*}
    where the inequality is since $\tilde{d}_j(v)$ approximates $E(v,V\setminus U(v))$ up to a factor of $2^5$.
    The overall probability of sampling an edge $e$ in one of the $\tau=A k n^{1+1/r} \log n$ samples satisfies,
    \begin{align*}
        p_e 
        &= 1-(1-1/(2^{10} m_j))^{\tau}
        \ge 1-\exp\left( \frac{A k n^{1+1/r} \log n}{2^{10} m_j} \right)
        \\
        &= \frac{A k n^{1+1/r} \log n}{2^{10} m_j} - O\left( \left( \frac{A k n^{1+1/r} \log n}{2^{10} m_j} \right)^2 \right)
        ,
    \end{align*}
    where the last inequality is from $(1-x/n)^n \le \exp(x)$ and the last equality is from the Taylor expansion of $\exp(x)$.
    Notice that the second element of the above is negligible whenever the algorithm does not sample all edges, in which case we are done.

    We will now show that the algorithm finds $k$-edge disjoint trees in any $\kappa=\Omega(m_j/n^{1+1/r})$-strong component of $G$.
    Let $C\subseteq V$ be a $\kappa$-strong component of $G$ and also $G'=(V,F)$ be the graph induced by the sample $F$.
    Finally, let $S\subseteq C$ be a cut such that $\mintcut_C(S)=\alpha\kappa$ for some $\alpha \ge 1$.
    For every $e\in E$ let $X_e$ be the random variable indicating whether $e$ was sampled into $F$.
    Notice that $\{X_e\}_{e\in E\cap (C\times C)}$ are negatively correlated Bernoulli random variables.

    Using the bound on $p_e$ we have that,
    \begin{equation*}
        \Exp{\sum_{e \in E(S,C\setminus S)} X_e}{} 
        = p_e \cdot |E(S,C\setminus S)|
        \ge \frac{A k n^{1+1/r} \log n}{2^{10} m_j} \cdot \frac{m_j}{n^{1+1/r} }
        \ge 100k\log n,
    \end{equation*}
    where the first inequality is since $|E(S,C\setminus S)| \ge m_j/n^{1+1/r}$ as  $C$ is $\kappa$-strong, and the bound on $p_e$, and the second is from setting $A$ appropriately.
    Applying the Chernoff bound for negatively correlated Bernoulli random variables (\Cref{lemma:negatively-correlated-chernoff}) we obtain that with probability at least $1-1/n^{7\alpha}$ the number of edges in $F\cap E(S,C\setminus S)$ is at least $k\log n$.
    We will now show that this holds for all cuts $S\subseteq C$.

    By \Cref{lemma:cut-counting} the number of cuts $S\subseteq C$ such that $\mintcut_C(S)\le\alpha\kappa$ is at most $|C|^{2\alpha}\le n^{2\alpha}$.
    Therefore, using a union bound over all cuts $S\subseteq C$,
    \begin{equation*}
        \Probability{\exists S\subseteq C \text{ s.t. } |E(S,C\setminus S)\cap F| < k\log n}
        \le \sum_{\alpha=1}^{\infty} \frac{1}{n^{7\alpha}} \cdot n^{2\alpha}
        \le n^{-5}.
    \end{equation*}
    Hence, all cuts $S\subseteq C$ have at least $k\log n$ edges sampled in $F$ with probability at least $1-n^{-5}$.
    Therefore, every two vertices $u,v\in C$ are connected by at least $k\log n$ disjoint paths in $G'[C]$, the sampled graph induced on $C$.
    This implies that $F$ contains $k$ disjoint spanning forests of $C$ and hence $C$ will be contracted in line~\ref{lst:line:contract-tree}.
    Using a union bound over all $O(n)$ $\kappa$-strong components of $G$ we have that with probability at least $1-n^{-4}$ all $\kappa$-strong components of $G$ are contracted.
    Since the algorithm contracts all $\Omega(m_j/n^{1+1/r})$-strong components of $G$, after the $j$-th iteration, there are no $m_j/n^{1+1/r}$-strong components in the graph.
    Using \Cref{lemma:edge-weight-connectivity} this implies that after contraction there remain at most $m_j/n^{1/r}$ edges.

    For the rest of the proof, assume that this event holds for all iterations of the main loop, and that the algorithm fails otherwise.
    After $r-1$ iterations of the main loop, at most $O(n^{1+1/r})$ edges remain in the graph.
    Then, in the sampling step of the $r$-th iteration, the algorithm samples $O(kn^{1+1/r}\log n)$ edges.
    Therefore, in expectation, each remaining edge is sampled $Ak\log n/2^5$ times.
    Hence, applying a Chernoff bound and setting $A$ appropriately we have that with probability at least $1-n^{-5}$, every edge is sampled at least $\log n$ times, and we can conclude that the algorithm finds all remaining edges.

    Finally, we show that if there were any edges left that can be added to $\{T_i\}_i$ without forming a cycle after $r-1$ iterations, then the algorithm will recover them in the last iteration.
    Assume that after $r-1$ iterations there exists some $(u,v)=e\in E\setminus(\cup_j T_j)$ and $T_i$ such that $e\cup T_i$ is a forest.
    Since $u,v$ are not connected in $T_i$, then the vertices $u,v$ are in different connected components, therefore $e$ was not contracted and will be recovered in the $r$-th iteration.
    Therefore, at the end of the algorithm there remain no edges that can be added to $\{T_i\}_i$ without forming a cycle and the algorithm returns a maximal $k$-packing of forests

    We now analyze the probability that the algorithm fails.
    Notice that we sample at most $O(rkn^{1+1/r}\log n)$ edges in total, using a union bound over all the application of \Cref{lemma:edge-sampling-1-round} we have that the all the sampling algorithms succeed with probability at least $1-rkn^{1+1/r-4}\log n\ge 1-n^{-2}$.
    In addition, we find that the decrease in edge count in each iteration happens with probability at least $1-n^{-4}$.
    Using a union bound over all iterations we have that the algorithm succeeds with probability at least $1-2/n^{-2}$.

    We conclude by bounding the query complexity of the algorithm.
    The query complexity of estimating the degree of a vertex is $O(\log n)$ queries in one round by \Cref{lemma:degree-estimation}.
    Since the algorithm estimates the degree of at most $n$ vertices during at most $r$ iterations the total query complexity for degree estimation is $O(rn\log n)$.
    Sampling a single edge requires $O(\log^3 n/\log \log n)$ non-adaptive queries by \Cref{lemma:edge-sampling-1-round}.
    The algorithm samples $O(kn^{1+1/r}\log n)$ edges in each round and therefore the overall query complexity of edge sampling is $\tO(rkn^{1+1/r})$.
    Therefore, the algorithm uses $\tO(rkn^{1+1/r})$ queries in $2r$ rounds.
\end{proof}

\section{Monotone Cost Matrix Problem}
\label{sec:monotone-cost-matrix}
In this section we show how to solve the monotone matrix problem using $O(rn^{1+1/r})$ queries in $r$ rounds, proving \Cref{lemma:monotone-cost-matrix}.

\begin{algorithm}[htbp]
    \caption{Recursive Monotone Matrix Procedure}
    \label{alg:monotone-cost-matrix}
    \begin{algorithmic}[1]
        \Procedure{MON-MAT}{$A, r$}
            \State \textbf{Input:} Monotone Cost Matrix $A\in \R_+^{a\times b}$, Parameter $r$.
            \State \textbf{Output:} Minimum in matrix.
            \If{$b=1$}
                \AdapRound{1}
                    \State \Return minimum of $A$. \label{lst:line:read-single-column}
                    
            \EndIf
            \State $m^* \gets \infty$
            \State $i_{s,0},i_{t,0}\gets 0$  
            \AdapRound{1}
            \State $\forall k\in[n^{1/r}], \quad i_{s,k},i_{t,k} \gets $ first and last indices of minimum in column $kb/n^{1/r}$

            \For{$k \in [n^{1/r}]$}
                \State $m_{k} \gets \text{MON-MAT}(A[i_{t,k-1},\ldots,i_{s,k}][(k-1)b/n^{1/r},\ldots, b/n^{1/r}],r)$
            \EndFor
            \State \Return $\min_k m_k$
        \EndProcedure
    \end{algorithmic}
\end{algorithm}

\begin{proof}
    The algorithm we use to solve the monotone matrix problem is presented in \Cref{alg:monotone-cost-matrix}.
    We begin by proving the correctness and then proceed to bound the query complexity.
    Notice that the termination condition of the recursion is when the matrix has a single column, and then returning the minimum value in that column. 
    Hence, the top level procedure returns the minimum over the minimum of all the columns.
    
    We now bound the query complexity.
    Notice that the recursion depth is $r$ since after $r$ rounds all the matrices have a single column.
    Since at each level we split the matrix into $n^{1/r}$ submatrices, at the $i$-th level of recursion we have the matrices $A_1,\ldots,A_{n^{i/r}}$.
    Denote the number of rows in $A_i$ by $a_i$.
    Notice that $\sum_i a_i \le n + n^{(i-1)/r}$ since for every column where $i_s=i_t$ we have a row that is added to both submatrices.
    Since we make $n^{1/r}$ queries in each matrix overall we require $n^{1/r}\sum_i a_i \le O(n^{1+1/r})$ queries in each level.
    Therefore, the algorithm requires $O(rn^{1+1/r})$ queries overall.
    Finally, notice that the algorithm runs in $r$ rounds since at every level it performs one set of non-adaptive queries and there are $r$ levels overall.
\end{proof}

\section{Other Cut Problems and Applications}
\label{sec:additional-results}
\subsection{\texorpdfstring{Minimum $s-t$ Cut and MAX-CUT estimation} {Minimum s-t Cut for Unweighted Graphs and MAX-CUT estimation}}
Solutions to both the minimum $(s,t)$ cut and the MAX-CUT estimation problems follow from a direct reduction to cut sparsifiers.

\begin{theorem}[$1-\epsilon$ MAX-CUT]
    \label{theorem:max-cut}
    Given a weighted graph $G$ on $n$ vertices with integer edge weights bounded by $W$, and a parameter $r\in \{1,2,\ldots,\log n\}$.
    It is possible to find a $1-\epsilon$ approximation of the maximum cut using $\tO(\epsilon^{-2} rn^{1+(1+\log W/\log n)/r})$ queries in $2r+1$ rounds.
    The algorithm is randomized and succeeds with probability $1-n^{-2}$.
\end{theorem}
\begin{proof}
    Begin by constructing a quality $(1\pm\delta)$-cut sparsifier $H$ using $\tO(rn^{1+(1+\log W/\log n)/r})$ queries in $2r$ rounds by \Cref{lemma:sparsifier-weighted-k-rounds}, this succeeds with probability $1-n^{-2}$.
    Then, find $S=\arg\max_{S\subseteq V} \mintcut_{H}(S)$ by checking all the possible cuts.
    Finally, return the value of the cut $\mintcut_{G}(S)$ using one additional query in one round.
    Notice that since $S$ is the maximum cut in $G'$ then,  
    \begin{equation*}
        \mintcut_G(S) 
        \ge (1-\delta)\mintcut_{H}(S) 
        \ge (1-\delta)\mintcut_{H}(S)
        \ge (1-\delta)\cdot (1-\delta)\mintcut_{G}(S)
        \ge (1-3\delta)\mintcut_{G}(S)
        .
    \end{equation*}
    Setting $\delta$ appropriately gives the desired result.
\end{proof}
The reduction to minimum $(s,t)$ cut from cut sparsifier is a bit more involved and uses the following result from \cite{RSW18}.
\begin{lemma}[Claim 5.3 in \cite{RSW18}]
    \label{lemma:sparsifier-to-min-s-t-cut}
    Let $G=(V,E)$ be an unweighted graph on $n$ vertices, let $s,t\in V$ be two vertices and let $H$ be a quality $(1\pm\epsilon)$ cut sparsifier of $G$ with $\epsilon=n^{-1/3}$.
    Furthermore, let $F$ be the maximal $s-t$ flow in $H$ and $H'$ be $H'-F$.
    Finally, obtain $G'$ from $G$ by contracting every $3\epsilon n$ connected components of $H'$.
    Then, $\mintcut_{G'}(s,t) = \mintcut_G(s,t)$ and the number of edges in $G'$ is at most $O(n^{5/3})$. 
\end{lemma}
We can now prove the following theorem.
\begin{theorem}[$(s,t)$ Minimum Cut]
    \label{theorem:s-t-minimum-cut}
    Let $G=(V,E)$ be an unweighted graph on $n$ vertices and let $s,t\in V$ be two distinct vertices.
    It is possible to find any the value of the minimum $s-t$ cut using $\tO(n^{5/3+1/r})$ queries in $2r+1$ rounds.
\end{theorem}
\begin{proof}
    Begin by constructing a quality $1+n^{-1/3}$ sparsifier $H$ using $\tO(n^{5/3+1/r})$ queries in $2r$ rounds by \Cref{lemma:sparsifier-weighted-k-rounds}.
    Then, obtain $G'$ according to the process described in \Cref{lemma:sparsifier-to-min-s-t-cut} and learn all the edges of $G'$ using $O(n^{5/3})$ queries by \Cref{corollary:graph-recovery-cut}.
    Finally, return the minimum cut of $G'$.

    The correctness of the algorithm is by \Cref{lemma:sparsifier-to-min-s-t-cut} and the query complexity is determined by the number of edges in $G'$ which again is determined by the same lemma.
    Finally, the success probability of the algorithm is $1-n^{-2}$ since the only randomized part of the algorithm is the construction of the cut sparsifier.
    This concludes the proof of \Cref{theorem:s-t-minimum-cut}.
\end{proof}

\subsection{Streaming Weighted Minimum Cut}
In this section we show how to apply our results for finding a minimum cut in a weighted graph in the streaming model.
In the model, we are given a graph $G$ as sequence of edges $(e_1,w_1),\ldots,(e_m,w_m)$ where $e_i$ is the edge, $w_i$ is the weight of the edge.
We allow the weights to be negative, to allow for deletions, and the edges to be repeated.
\begin{corollary}
    \label{corollary:weighted-min-cut-streaming}
    Let $G$ be a weighted graph given as a dynamic stream $(e_1,w_1),\ldots,(e_m,w_m)$, and a parameter $r\in \{1,2,\ldots,\log n\}$,
    it is possible to find a minimum cut of $G$ using $\tO(rn^{1+1/r})$ storage complexity in $r+2$ passes.
\end{corollary}
\begin{proof}
    The proof will follow from the reduction provided in \Cref{lemma:mn20-summary}, similarly to the proof of \Cref{theorem:weighted-min-cut}, which shows that the problem of finding a minimum cut in a weighted graph can be reduced to the problem of finding a cut sparsifier.
    The first step, constructing a cut sparsifier, can be done in the streaming model using $\tO(\epsilon^{-2} n)$ space in a single pass \cite{KLMMS17}.
    The second step, solving the monotone matrix problem, can be achieved using our algorithm for the problem \Cref{lemma:monotone-cost-matrix}.
    Notice that each cut query can be simulated in a single pass of the stream using $O(\log n)$ space, simply by summing the weights of the edges in the cut.
    Therefore, each round of the cut query algorithm can be simulated in a single pass of the stream.

    Combining the two steps, there exists an algorithm that can find a minimum cut of a weighted graph using $\tO(rn^{1+1/r})$ space in $r+2$ passes.
\end{proof}

\section{\texorpdfstring{$2$-Respecting Min-Cut Reduction}{2-Respecting Min-Cut Reduction}}
\label{sec:2-resp-min-cut-reduction}
In this section we prove \Cref{lemma:mn20-summary}.
For ease of reference we begin by restating the lemma.

\mnsummary*

The proof is based on analyzing the minimum cut algorithm proposed in \cite{MN20}, which is based on the 2-respecting min-cut framework.
An outline of the algorithm is shown in \Cref{alg:mn20-summary}.
\begin{algorithm}[htbp]
    \caption{Minimum Cut Algorithm of \cite{MN20}}
    \label{alg:mn20-summary}
    \begin{algorithmic}[1]
        \State \textbf{Input:} Weighted graph $G=(V,E,w)$.
        \State \textbf{Output:} Minimum in matrix.
        \State $H\gets$ cut sparsifier of $G$.
        \State $T_1,\ldots,T_k\gets $ pack $\Theta(\log n)$ trees of $H$.
        \Comment{$k=\Theta(\log n)$.}
        \For {$i\in [k]$}
            \State $m_{a} \gets$ minimum of $1$-respecting minimum cuts of $T_i$
            \State $\mathcal{P}\gets$ heavy-light decomposition of $T$
            \For {$\pi\in \mathcal{P}$}
                \State $m_{\pi}\gets$ minimum of $2$-respecting minimum cuts of $T'=\pi$
            \EndFor
            \State $m_b \gets \min_{\pi} m_{\pi}$.
            \State $S \gets \{(\pi_1,\pi_2) \in \mathcal{P}^2 \mid (\pi_1,\pi_2) \text{ relevant for minimum cut}\}$ \label{lst:line:relevant-pairs}
            \For {$(\pi_1,\pi_2)\in S$}
                \State $m_{\pi_1,\pi_2}\gets$ minimum of $2$-respecting minimum cuts of $T'=\pi_1\cup\pi_2$
            \EndFor
            \State $m_c \gets \min_{\pi_1,\pi_2} m_{\pi_1,\pi_2}$
            \State $m_i\gets \min(m_a,m_b,m_c)$
        \EndFor
        \State \Return $\min_{i\in [k]} m_i$
    \end{algorithmic}
\end{algorithm}
The main result of \cite{MN20} is the following, on which we base the proof of \Cref{lemma:mn20-summary}.
\begin{theorem}
    \label{theorem:mn20}
    Let $G=(V,E,w)$ be a weighted graph, then \Cref{alg:mn20-summary} returns a minimum cut of $G$.
\end{theorem}
We also require the following results from \cite{MN20}.
\begin{lemma}[Claim 3.1, Algorithm 3.3 from \cite{MN20}]
    \label{lemma:mn20-summary-num-2-respecting}
    Each $2$-respecting minimum cut problem on a path of length $a$, can be solved reduced to $O(\log n)$ instances of the monotone matrix problem of size $O(a)$, that can be solved in parallel.
\end{lemma}
\begin{lemma}[Claim 3.14 of \cite{MN20}]
    \label{lemma:mn20-summary-2-respecting}
    Algorithm \ref{alg:mn20-summary} solves the $2$-respecting minimum cut problem on paths of total size at most $\tO(n)$, each of maximum length at most $n$.
\end{lemma}
\begin{lemma}[Lemma 5.5 from \cite{MN20}]
    \label{lemma:mn20-summary-find-s}
    It is possible to find the set $S$ in line~\ref{lst:line:relevant-pairs} using $\tO(n)$ non-adaptive cut queries.
\end{lemma}

We are now ready to prove \Cref{lemma:mn20-summary}.
\begin{proof}[Proof of \Cref{lemma:mn20-summary}]
    The algorithm begins by constructing a cut sparsifier of $G$ using $\tO(q_S)$ queries in $r_S$ rounds.
    Next, the tree packing step can be performed without requiring any additional queries.

    We now continue to analyzing the query complexity and round complexity of the main loop of the algorithm.
    Observe that all main loop iteration can be performed in parallel, and hence we only need to bound the round complexity of a single iteration.
    Furthermore, performing the heavy-light decomposition step depends only on the tree $T_i$ and does not require any additional queries.
    Therefore, it remains to analyze the query complexity of calculating $m_a,m_b,m_c$.

    Given a tree $T_i$, a $1$-respecting minimum-cut is simply a partition of the tree into two sets of vertices.
    Hence, the algorithm can compute $m_a$ using $O(n)$ non-adaptive queries.
    Next, notice that calculate $m_b,m_c$ reduces solving instances of the $2$-respecting minimum-cut problem on paths of total length at most $\tO(n)$ by \Cref{lemma:mn20-summary-2-respecting}.
    To determine these instances, the algorithm needs to recover the set $S$;
    this requires $\tO(n)$ non-adaptive queries by \Cref{lemma:mn20-summary-find-s}.

    We now analyze the query complexity of calculating $m_b,m_c$.
    Let $P_1,P_2,\ldots$ be the set of paths on which the algorithm performs the $2$-respecting minimum-cut.
    Solving the $2$-respecting minimum-cut problem on a path of length $a$ is equivalent to solving $O(\log n)$ instances of the monotone matrix problem on a matrix of size $O(a)$ by \Cref{lemma:mn20-summary-num-2-respecting}.
    Therefore, the query complexity of solving the $2$-respecting minimum-cut problems on all the paths is $O(\log n)\cdot \sum_i q_M(|P_i|)$, where $q_M$ is the query complexity of the monotone matrix problem. 
    Notice that when $q_M(n)$ is a convex function we can bound the total query complexity by $q_M(n)\cdot \polylog n$, since $|P_i|\le n$ by \Cref{lemma:mn20-summary-2-respecting}.
    Therefore, the algorithm computes $m_b,m_c$ using $O(\polylog n)\cdot q_M(n)$ queries in $r_M(n)+1$ rounds, where the additional round is from calculating $S$.

    We conclude the proof by bounding the round complexity.
    The algorithm requires $r_S$ rounds to construct the cut sparsifier, $r_M(n)+1$ rounds to compute $m_b,m_c$.
    Notice that it also requires one round to compute $m_a$, however we can perform it in parallel to calculating $m_b,m_c$.
    In conclusion, $r_S+r_M(n)+1$ rounds.
    To analyze the query complexity, simply sum the query complexity of all the steps which yields the desired bound.
\end{proof}

\section{\texorpdfstring{Unweighted Minimum Cut in $2r+1$ Rounds through Cut Sparsification} {Unweighted Minimum Cut in 2r+1 Rounds through Cut Sparsification}}
\label{sec:min-cut-sparsifier-k-rounds}
In this section we present an algorithm for finding the global minimum cut of a graph by constructing cut sparsifiers with low adaptivity.
\begin{theorem}[Unweighted Minimum Cut with $O(r)$ Rounds Through Cut Sparsifiers]
    Given an unweighted graph $G$ on $n$ vertices and a parameter $r\in \{1,2,\ldots,\log n\}$, it is possible to find a minimum cut of $G$ using $\tO(rn^{1+1/r})$ cut queries in $3r+4$ rounds.
    The algorithm is randomized and succeeds with probability $1-n^{-2}$.
\end{theorem}
The algorithm is based on the following reduction, showing that given a cut sparsifier of $G$ we can find a small set of edges that cover all small cuts.
\begin{lemma}[Lemma 2.6 from \cite{RSW18}]
    \label{lemma:small-cut-cover-RSW}
    Let $G=(V,E)$ be a graph with minimum cut $\lambda$ and minimum degree $\delta(G)$.
    Then, the number of edges that participate in any cut of size at most $c+0.8\cdot\delta(G)$ is $O(n)$.
\end{lemma}
Hence, given a sparsifier it is possible to contract all edges that do not participate in cuts of size at most $c+1.5\delta(G)$.
This leaves a sparse graph which we can learn efficiently with adaptivity $0$ using \Cref{corollary:graph-recovery-cut}.
We are now ready to prove the theorem.
\begin{proof}
    Begin by constructing a quality $51/50$ sparsifier $G'$ using \Cref{lemma:sparsifier-weighted-k-rounds} and querying the degree of every vertex.
    If no non-trivial cut is smaller than $51 \delta(G)/50$ then we return a minimum degree vertex.
    Otherwise, we contract every edge in that is not part of a cut in $G'$ of size at most $53 \lambda(G')/50$.
    Let $H$ be the graph obtained by contracting all the edges.

    Notice that since $\lambda(G) \le 51 \delta(G)/50$ every cut that is of size at most $53 \lambda(G')/50$ has size in $G$ of at most,
    \begin{equation*}
        53 \lambda(G')/50 
        \le (3/2) \cdot 53 \lambda(G)/50
        \le \lambda(G) + 51/50((53/50)(3/2)-1)\delta(G)
        < \lambda(G) + 0.8\delta(G) 
        .
    \end{equation*}
    Therefore, the number of edges in $H$ is at most $O(n)$ by \Cref{lemma:small-cut-cover-RSW}.
    Furthermore, since $G'$ is quality $51/50$ cut sparsifier of $G$ then every minimum cut of $G$ has value $\le 53/50\lambda(G')$ in $G'$.
    Hence, the minimum cut of $H$ is the same as the minimum cut of $G$.
    Using \Cref{corollary:graph-recovery-cut} we can then learn $H$ using $O(n)$ queries and find the minimum cut of $G$.

    To analyze the query complexity of the algorithm observe that by \Cref{lemma:sparsifier-weighted-k-rounds} we can construct the sparsifier in $\tO(rn^{1+1/r})$ queries in $3r+3$ rounds.
    We then require an extra round to learn the graph.
    Hence, overall we require $\tO(rn^{1+1/r})$ queries in $3r+4$ rounds.

    We conclude by bounding the probability of failure, notice that the only randomized part of the algorithm is the construction of the sparsifier.
    Since it succeeds with probability $1-n^{-2}$, then the overall probability of failure is also $n^{-2}$.
    This concludes the proof of \Cref{theorem:min-cut-k-rounds}.
\end{proof}

\end{document}